\documentclass[runningheads]{llncs}
\usepackage[T1]{fontenc}
\usepackage{graphicx}
\usepackage{graphicx} 
\usepackage{color}
\usepackage{url} 
\usepackage{amsmath, amssymb}
\usepackage{tikz}
\usepackage{mathtools}
\usepackage{graphicx}
\usepackage{subcaption}
\usepackage[dvipsnames]{xcolor}
\usepackage{multicol}
\usepackage{fontawesome5}
\usepackage{cite}
\usepackage{hyperref}
\usepackage{istgame}
\usepackage{breqn}
\usetikzlibrary{shapes.geometric, arrows.meta, positioning}
\usepackage{xspace}

\usepackage{algorithm}
\usepackage{algorithmicx}
\usepackage{algpseudocode}
\usepackage{hyperref}
\usepackage{booktabs}
\usepackage{rotating}
\usepackage{multirow}
\usepackage{comment}
\usepackage{cleveref}
\usepackage{tikz}
\usetikzlibrary{shapes.geometric, arrows.meta, positioning}

\tikzset{
  block/.style = {rectangle, draw, thick, minimum width=2cm, minimum height=1cm, text centered, align=center},
  decision/.style = {diamond, draw, thick, minimum width=2cm, minimum height=1cm, text centered, aspect=1.2},
  arrow/.style = {-{Stealth[length=6pt]}, thick},
}

\usepackage[normalem]{ulem}
\newcommand{\lit}[1]{\text{#1}}   
    
\newcommand{\Part}[2]{\State \textbf{#1} #2}
\newcommand{\EndPart}{}

\newcommand{\trap}{{\sc Trap}\xspace}

\newcommand{\protocol}{{\sc Stake}\xspace}

\newcommand{\mypar}[1]{\paragraph{#1.}}

\newcommand{\ms}[1]{%
    \relax\ifmmode
        \mathord{\mathcode`\-="702D\it #1\mathcode`\-="2200}%
    \else
        {\it #1}%
    \fi
}

\newcommand{\tup}[1]{%
    \relax\ifmmode
      \langle #1 \rangle%
    \else
        $\langle$ #1 $\rangle$%
    \fi
}

\newcommand{\nradii}[2]{%
  \pgfmathsetmacro{\anglestep}{360/#1}%
  \foreach \i in {1,...,#1}{%
    \pgfmathsetmacro{\ang}{90 + (\i-1)*\anglestep}%
    \draw (0,0) -- (\ang:#2);%
  }%
}
\algnewcommand\algorithmicwaituntil{\textbf{wait until}}
\algnewcommand\WaitUntil{\State \algorithmicwaituntil{}}

\title{{\sc{Stake:}} 
Preventing a colluding majority from double spending
}

\begin{document}


\author{Zahra Naderi\inst{1} \and Vincent Gramoli\inst{2}}

\institute{
University of Sydney \\
\email{zahra.naderi@sydney.edu.au}
\and
University of Sydney and Redbelly Network \\
\email{vincent.gramoli@sydney.edu.au}
}


\maketitle

\begin{abstract}
Blockchains typically operate in open networks where it is hard to control the delay of messages. While classic blockchain protocols assumed synchrony, newer protocols try to remain secure despite unexpected network delays. Unfortunately, in the traditional model where $n$ participants can either be honest or Byzantine, we need a supermajority, or  
$\lfloor 2n/3\rfloor + 1$, of them to be honest. 
Recent work have explored the question of reducing the number of needed honest nodes in a game theory model to a majority sometimes by introducing $k$ rational players and $t$ Byzantine players. Yet, to our knowledge, no blockchain protocol managed to reduce it further. 

In this paper, we offer a blockchain protocol called \emph{Secure and Tolerant Algorithm through (K,t)-robust Equilibrium (\protocol)} that works with only $\lfloor n/3\rfloor + 1$ honest players. \protocol 
requires each consensus participant to stake a sufficiently large amount $s$ compared to their liquidity $\ell$. We show that the probability that a coalition manages to execute $\zeta$ double spending attacks (we call them a $\zeta$-uple attack) drops exponentially fast with $\zeta$. This allows us to demonstrate that \protocol reaches a ($k,t$)-robust equilibrium, or that rational players do not collude with Byzantine, even with a relatively low $s/\ell$ ratio.
\end{abstract}

\section{Introduction}

Blockchains typically operate in open networks where it is hard to control the delay of messages. While classic blockchain protocols~\cite{Nak08,wood2014ethereum,GKL24} assume synchrony~\cite{dwork1988consensus} or that all messages take less than a known delay  to reach their destination, they can be subject to double spending in case of network delays, as it sometimes happened~\cite{doublespending}.

To be secure against unexpected message delays, some blockchains~\cite{CNG21,buchman2018latest} tolerate partial synchrony~\cite{dwork1988consensus}.
The typical idea is to prevent forks by having $n$ participants reaching consensus upon a unique block before appending it to the chain~\cite{CNG21,buchman2018latest}.
As long as the set of $n$ participants is an unpredictable subset of all participants,
this method can be used in a permissionless environment.
For this solution to work though, more than two thirds of the participants should be honest~\cite{LSP82}.
Unfortunately, in a blockchain environment where participants can steal large amount of cryptocurrency by misbehaving it is unrealistic to assume that at least
$\lfloor 2n/3\rfloor + 1$  are honest.

A more realistic model is to consider the blockchain as a repeated consensus game where $k$ players are rational~\cite{Abraham2016SolidusAI,ranchal2022trap,avarikioti2025cobra}, either incentivized to behave or to collude depending on the corresponding incentives, and where $t$ players are Byzantine while the remaining $n-k-t$ players are honest.
Solving the blockchain problem in this model consists of designing an algorithm that defines a $(k,t)$-robust equilibrium~\cite{DBLP:conf/tcc/AbrahamDH08}, where $k$ rational players cannot increase their utility by colluding with $t$ Byzantine players in consecutive consensus executions.

The key idea is then to introduce incentives in the protocol, just like when proof-of-stake blockchains slash the stake of misbehaving participants. Solidus~\cite{Abraham2016SolidusAI} reaches a $(k,t)$-robust equilibrium, unfortunately it still requires more than 2/3 of the players to be honest. 
{\sc Trap}~\cite{ranchal2022trap} baits rational players to betray colluding players to only require more than half of the participants to be honest.
Unfortunately, it only solves the one-shot consensus problem and was not applied to a repeated game where players could adapt their strategy during a consensus instance based on their observations of the previous consensus instance. 

\mypar{Contributions}
In this paper, we offer the first blockchain protocol, called \emph{Secure and Tolerant Algorithm through (K,t)-robust Equilibrium (\protocol)}, that works with $\lfloor n/3\rfloor + 1$ of honest players.
Our result works in partial synchrony, which was originally believed to be impossible~\cite{avarikioti2025cobra}. By contrast with this impossibility result, our algorithm requires all consensus participants to stake a small but sufficiently large amount $s$ compared to their liquidity $\ell$. 
%
%
Note that this requirement applies well to modern regulated public blockchains 
like 
Avalanche~\cite{rocket_scalable_2020} Evergreen subnetwork, Concordium~\cite{concordium}, Mina~\cite{cryptoeprint:2020/352} and
Redbelly~\cite{CNG21}, where validators must meet non-discretionary rules to create accounts.

These requirements alone would not, however, solve this difficult problem.
Indeed, if rational players can execute an infinite amount of double spending attacks by colluding with Byzantine players, then they would always be incenvitized to do so to maximize their reward.
Interestingly, we show that the probability for a coalition to execute $\zeta$ double spending attacks, or a so-called \emph{$\zeta$-uple spending attack}, drops exponentially fast with the number $\zeta$ of carefully selected partitions. In short, the proof follows from the growing difficulty of partitioning honest players from each other as the coalition multiply the number of partitions and is confirmed empirically by the rarity of successful double spending attacks in classic blockchains, like Ethereum~\cite{wood2014ethereum}.

To demonstrate that \protocol reaches a ($k,t$)-robust equilibrium, we formalise the problem as an imperfect information game where players are Byzantine, honest or rational. We build upon recent work on the accountable consensus~\cite{DBLP:conf/icdcs/CivitGG21} and baiting strategy~\cite{ranchal2022trap} to define a protocol where rational players are incentivized to betray a coalition of attackers. 
Ensuring that \protocol solves the blockchain problem is 
non-trivial as the Folk theorem~\cite{31353} prevents us from converting a one-shot consensus game into its corresponding repeated game necessary for a (multi-shot consensus) blockchain.
The reason is that an equilibrium in an incomplete information game does not translate into an equilibrium in the corresponding repeated game, where players learn information from past game instances that impact their upcoming strategies.

The paper is organized as follows. Section~\ref{sec:model} presents the model. Section~\ref{sec:problem} introduces the problem of reaching a ($k,t$)-robust equilibrium despite $\zeta$-uple spending attempts. Section~\ref{sec:solution} introduces the \protocol protocol that solves the problem. 
Section~\ref{sec:Solution} explains why conducting a successful $\zeta$-uple spending attack is non trivial and proves that the proposed protocol solves the problem whether or not Byzantine players are able to bribe rational players.
Section~\ref{sec:analysis} illustrates how to select a stake to secure a blockchain based on the proposed protocol.
Section~\ref{sec:rw} presents the related work and Section~\ref{sec:conclusion} concludes.

\section{Model}\label{sec:model}

We consider a game among a set $N$ of $n$ players communicating in a partial synchrony distributed system, where there is an unknown bound on the delay of messages~\cite{dwork1988consensus}. The set of possible player types is $\Theta = \times_{i \in N} \Theta_i$ where $\Theta_i=\{\ms{Byzantine}, \ms{Rational}, \ms{Honest}\}$ and $\theta_i \in \Theta_i$. The \textit{honest} players always follow the protocol; the set $T\subseteq N$ of $t \leq t_0$ Byzantine players, where $t_0=\lceil \frac{n}{3} \rceil -1$, can try to initiate an attack regardless of apparent profitability and the set $K\subseteq N$ of $k$ rational players, where $t+k<n-t_0$, try to maximize their expected utility. Note that the bound on $t$ is needed to solve consensus~\cite{LSP82} in the general setting and if $t\geq n-t_0$ then Byzantine players can act as shadow authorities attempting to replace the consensus algorithm with their own~\cite{DBLP:conf/icdcs/CivitGG21}.
\emph{Nature} refers to a set of non-strategic players with a predetermined and exogenous stochastic strategy, where their decision is independent from  the strategy of other players~\cite{tadelis2013game}. 
We thus refer to \emph{Nature} as a coalition of Byzantine players that control the scheduler.

\mypar{Adversary}
To model the adversary, we define two types of schedulers that run sequentially, each managing a global asynchronous round, 
which is split into several asynchronous local rounds. In the first global round, whose duration equals a Global Stabilization Time (GST)~\cite{dwork1988consensus}, players are scheduled to move at arbitrary times by an adversary-controlled scheduler. When they are scheduled, they receive the unsent messages from the last local round, perform some internal computation on them, and then send some messages at the start of the next local round. Each player's scheduler works independently of the others, and the adversary controls it. However, the adversary cannot indefinitely obstruct the periodic scheduler; its control over the scheduler is limited to the initial global round, and when GST occurs, it is unable to extend its influence on it. After GST, players proceed in locksteps such that each message is sent in a local round and in the subsequent local round, this message is received, computation is performed and new messages are potentially sent, which is similar to how Abraham et al.~\cite{DBLP:journals/teco/AbrahamDH19} define scheduler in asynchronous environment.
To compute the utility, we need a probability distribution over outcomes that depends on the order in which players are scheduled and on message-delivery times, both of which the adversary can control prior to the occurrence of GST.
Moreover, we assume Byzantine players know all players' types. We also assume no transparency in the network, meaning other players cannot detect whether multiple accounts are controlled by the same player.

Furthermore, we consider a slowly-adaptive adversary, whereby the number of attackers and the strategies employed by Byzantine players can progressively (rather than instantaneously) change, but only between local rounds, not during a local round.

\mypar{Extensive game}
We define an \emph{extensive game} $\Gamma$ as a tuple $\langle N, H,\Theta, u \rangle$ over the set $N$ of $n$ players and where $H$ is the \emph{set of histories}. 
We assume \emph{imperfect information} in that players are uncertain about other players and their respective actions and types, and \emph{perfect recall} in that  players never forget information they once knew or actions they have previously chosen.
More formally, each component of a history is an action taken by a player, generated by the sequence of actions $\{Attack, Pass\}$, $\{Bribe, \neg Bribe\}$, $\{Betray, Collude\}$ and $\{Attack, Commit\}$. The full formal construction of $H$ and the game tree of $\Gamma$ is provided in Appendix~\ref{app:gameTree}.
This definition is classic in game theory~\cite{osborne1994course, tadelis2013game}. 

\mypar{Strategy}
A \textit{strategy} $\sigma_i$ maps player $i$'s view (its own type, past actions and the messages it has received) to an action (which is a move a player can make at a stage in a game). 
Let $u_i(\theta, \sigma)$ be the expected utility of rational player $i$ under type profile $\theta$ and each player $j \in N$ is playing the strategy $\sigma_j$. We assume $u_i$ depends on $\theta$ only through the number of players of each type.

\mypar{Cryptography}
We assume that Byzantine and rational players are polynomially bounded 
and we reuse the enhanced trapdoor permutation cryptography assumption from~\cite{DBLP:books/acm/19/0001MW19a}. Informally, an oblivious transfer is a protocol in which a sender sends multiple messages and never learns which one the receiver chooses. In practice, these two cryptographic assumptions imply that players can sign unforgeable messages and perform an oblivious transfer. We assume a PKI (Public Key Infrastructure) where each player is equipped with an unforgeable private key and a public key, with all public keys being common knowledge. These assumptions allow players to verify other players' moves without revealing their own strategies.

Finally, we assume the messages employ \textit{transferable authentication}, meaning all correct players can verify the identity of the message generator, even if they receive the message indirectly as defined in~\cite{clement2012limited}. 

\mypar{Coalitions}
After the Byzantine players make their decision to attack, they send a Byzantine coalition formation call to the rational players; as a result, rational players can distinguish which information set they are in. Then, they may bribe rational players with the profit they steal to incentivize rational players to join their coalition.
Rational players simultaneously decide whether to betray or join the coalition, without knowing the actions of other players. If they choose to betray to maximize their payoffs, even Byzantine players will not know whether rational players will betray. 

\mypar{Attack}
A \emph{blockchain} is a sequence of blocks. A blockchain protocol implements the blockchain with a distributed protocol and may lead this blockchain to fork or create \emph{branches} when multiple blocks are created at the same \emph{index} of the chain. One cause of such a fork is a double-spending attack or more generally \emph{$\zeta$-uple spending attack} with $\zeta \geq 2$, where attackers try to send the same coins $\zeta$ times. Note that the transferable authentication property makes an attacker \textit{detectable} when the attacker equivocates while $\zeta$-uple spending with $\zeta \geq 2$.

\mypar{Stake and liquidity}
We assume that all players split their assets into two parts: \emph{stake} $s$ and \emph{liquidity} $\ell$. By staking some of their assets, they are allowed to participate in the consensus algorithm as a player 
and take a decision or vote. 
Players can transfer their liquidity by issuing transactions
but we omit the transaction fees for the sake of simplicity.
If a player manages to successfully execute a $\zeta$-uple spending attack where $\zeta \geq 2$, then he could effectively multiply its liquidity by $\zeta$.
We consider that the quantity of staking and liquidity is constant ($r=\frac{s}{\ell}$).

Note that modern blockchains 
requiring some users to meet non-discretionary rules~\cite{concordium,cryptoeprint:2020/352,CNG21} can automate this check via smart contracts before letting users become players.

\section{The Problem of Coping with $\zeta$-uple Spending Attacks}\label{sec:problem}

Our goal is to design a  blockchain protocol 
with a minimal stake requirement to disincentivize players from double-spending or more generally $\zeta$-uple spending with $\zeta \geq 2$. Specifically, the goal is to
require $n$ players to stake an amount $s$ of coins to reach an $\epsilon$-$(k,t)$-robust equilibrium despite the presence of $t$ Byzantine, among which $f\leq t$ are detectable attackers and $k$ are rational players. 
While anyone can join the blockchain protocol after staking 
an amount $s$ of coins out the total $s+\ell$ that it owns,
only $n$ of these players are selected to run a consensus protocol to agree
on a unique block to append to the blockchain. 
Note that restricting consensus execution to $n$ players is easily done by listing the set of consensus participants in blocks~\cite{CNG21} and used in modern Byzantine fault tolerant blockchains~\cite{the_diem_team_diembft_2021,noauthor_aptos_2022,TG24}.

\mypar{The $\zeta$-uple spending attack}
We define a $\zeta$-uple spending attack, where $\zeta \geq 2$, as a generalisation of the double spending attack: While in a double-spending attack, \emph{attackers} typically create 2 branches to gain an amount as large as their liquidity, a $\zeta$-uple spending attack requires attackers to create $\zeta$ branches to gain an amount up to $\zeta-1$ times as large as their liquidity. 

For example, a $4$-uple spending attack can be successful after two consecutive successful double spending attacks: one that multiplies the initial liquidity $\ell$ by two, and a subsequent one that multiplies the resulting liquidity $2\cdot \ell$ by 2, leading effectively to quadrupling the initial liquidity $\ell$. In this case, the two consecutive double spending attacks occur at different indices of the chain.

%
We consider that honest players are uniformly distributed in the network, so we can represent the network as a disk. Attackers can create $\zeta$ branches by executing $\zeta$ successful \emph{isolation attacks} or delaying messages on the communication channels crossing $\zeta$ radii of the disk as depicted in Figure~\ref{fig:fig5}.
\begin{figure}[H] 
\begin{center}
\scalebox{0.5}{
\begin{tikzpicture}[line width=1.4pt]
 
\def\R{1.1}
\def\hsep{3.6}
 
\begin{scope}[xshift=0cm]
  \draw (0,0) circle (\R cm);
  \draw (0,0) -- (180:\R);
  \draw (0,0) -- (0:\R);
  \node[font=\huge] at (0, -1.7) {$\zeta=2$};
\end{scope}
 
\begin{scope}[xshift=\hsep cm]
  \draw (0,0) circle (\R cm);
  \nradii{3}{\R}
  \node[font=\huge] at (0, -1.7) {$\zeta=3$};
\end{scope}
 
\begin{scope}[xshift=2*\hsep cm]
  \draw (0,0) circle (\R cm);
  \nradii{4}{\R}
  \node[font=\huge] at (0, -1.7) {$\zeta=4$};
\end{scope}
 
\begin{scope}[xshift=3*\hsep cm]
  \draw (0,0) circle (\R cm);
  \nradii{5}{\R}
  \node[font=\huge] at (0, -1.7) {$\zeta=5$};
\end{scope}
 
\end{tikzpicture}
}
\caption{Increasing structural complexity of $\zeta$ isolation attacks resulting in $\zeta$ branches}
\label{fig:fig5}
\end{center}
\end{figure}
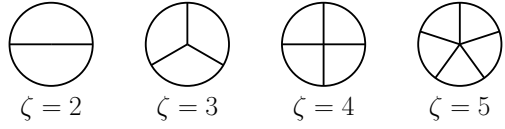

The key to avoid $\zeta$-uple spending attacks is by incentivizing rational players to betray attackers rather than colluding with them
by threatening to slash their stake $s$.
The following theorem gives the 
the minimum number of rational players that should betray the coalition for honest players to be informed of an attack.
Second, it gives 
a sufficiently large number of players from which a common message should be received.
The proof is deferred to Appendix~\ref{app:the2}.

\begin{theorem}\label{Theorem:the2}
If $t + k$ attackers attempt to carry out a $\zeta$-uple spending attack, and $\lfloor \frac{k+t}{2}-\frac{n}{6} \rfloor +1$ rational players are randomly selected from a set of $k$ rational players and they send message $m'$, then all honest players receive $m'$.
\end{theorem}

Ideally, rational players must find behaving more profitable than attacking ($\Gamma$ with $h_0 = \{Pass\}$). In addition, when Byzantine players decide to start an attack, betrayal should become more profitable than joining the coalition ($\Gamma$ with $h_0 = \{Attack\}$). As a result, betrayal and behaving should be two Nash equilibria.
Nevertheless, a Nash equilibrium only prevents a single rational player from increasing its utility by deviating (changing strategy) alone. Although necessary, this condition is not sufficient because we also need to prevent a coalition of players from deviating simultaneously. We discuss other variants of the Nash equilibrium in Section~\ref{sec:rw}.\\

Next, we need to introduce an equilibrium definition that allows to define a protocol
where rational players cannot gain more than a negligible amount $\epsilon$ by joining the attacking coalition rather than betraying it. This equilibrium is called the $\epsilon$-$(k,t)$-robust equilibrium~\cite{DBLP:conf/tcc/AbrahamDH08}.

Consider $\tau_{K \cup T}$ as the deviation strategy for the entire coalition, $\sigma_{K}$ as the betrayal strategy for $k$ rational players, $\tau_{T}$ as the deviation strategy for $t$ Byzantine players and $\rho_{(K \cup T)^{-}}$ as the strategy of honestly following the protocol for players who are not part of the coalition.

\begin{definition}[$\boldsymbol{\epsilon}$-{\boldmath $(k,t)$}-robust equilibrium]\label{def:def1}
If $\epsilon \geq 0$, then $\sigma$ is an $\epsilon$-$(k,t)$-robust equilibrium in game $\Gamma$ and $h_0=\{Attack\}$ if
\begin{equation} \label{eq:eq1}
u_i\left(\theta, \tau_{K \cup T}, \rho_{(K \cup T)^{-}}\right) \leq u_i\left(\theta, \sigma_{K}, \tau_T, \rho_{(K \cup T)^{-}}\right) + \epsilon.
\end{equation}
\end{definition}

Definition~\ref{def:def1} offers two interesting properties, the $k$-resilience and $t$-immunity.
First, $k$-resilience means that no coalition of up to $k$ rational players can increase their payoffs by more than a negligible amount by deviating from the protocol, even if they communicate with up to $t$ Byzantine players who help them and have knowledge of each other's types. Second, $t$-immunity refers to resilience against up to $t$ Byzantine players whose utilities are unknown. Regardless of the actions or type revelations of these $t$ players, the payoffs of the honest players should remain unaffected. Even if the Byzantine players communicate and coordinate, they should not be able to reduce the payoffs of others. Alongside this, note that a $(1, 0)$-robust equilibrium is a Nash equilibrium and a $(k,t)$-robust equilibrium is a $0$-$(k,t)$-robust equilibrium.

\section{\protocol: A Solution to the $\zeta$-uple Spending Attack}\label{sec:solution}
Our solution is the \protocol protocol (Secure and Tolerant Algorithm through (K,t)-robust Equilibrium), 
which implements a blockchain and prevents $\zeta$-uple spending attacks.  
If Byzantine players desire to establish an attack and form a coalition, rational players inside the coalition have two options: stay loyal to the coalition and gain the attack's profit; or betray, gain the protocol incentive, and alert honest players of the imminent attack, enabling them to stop it before it is successfully completed.

\begin{figure}
\includegraphics[width=\textwidth]{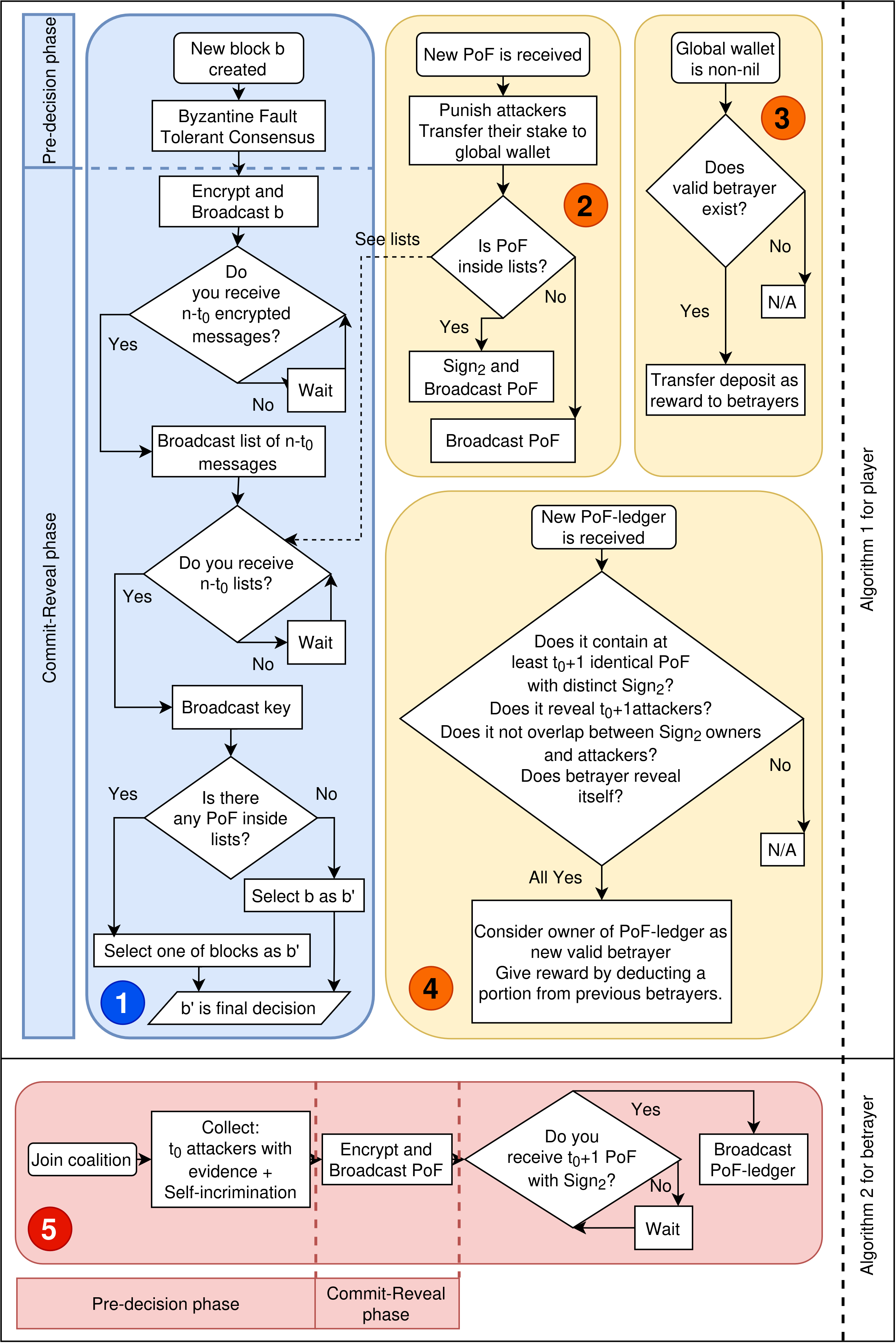}
\caption{The \protocol protocol consists of Algorithm 1 followed by any players and Algorithm 2 followed by betrayers. The pseudocode of these two algorithms is deferred to Alg.~\ref{alg:P} and~\ref{alg:P_betray} of the Appendix~\ref{app:pseudocode}.}
\label{fig:fig4}
\end{figure}

\protocol comprises different components illustrated in Figure~\ref{fig:fig4}. First, the algorithm executed by honest players is indicated in Flowcharts 1-4 of Figure~\ref{fig:fig4}. Second, the algorithm that rational may execute for maximizing their utility by betraying is depicted in Flowchart 5. In particular, Flowchart 1 describes the steps taken by a honest player each time a new block $b$ is proposed in order to reach consensus. Note that even though the first step is to invoke a Byzantine fault tolerance consensus protocol, Flowchart 1 continues in a series of steps to guarantee that attackers are identified. As a result, at the time the final decision,  block $b'$, is taken, the attackers will be identified (as we will show in Theorem~\ref{Theorem:the1}). This process is repeated for each index of the chain.

More specifically, as depicted in Flowcharts 1 and 5 of Figure~\ref{fig:fig4} and as detailed in the pseudocode of Appendix~\ref{app:pseudocode}, the \protocol protocol proceeds in two phases:

\begin{enumerate}
\item {\bf Pre-decision phase.}
In the pre-decision phase, all players execute a Byzantine fault-tolerant consensus to either agree on a block, denoted as $b$, or for Byzantine and rational players to communicate to form a coalition and cause a disagreement via a $\zeta$-uple spending attack. This consensus can be any arbitrary Byzantine fault-tolerant consensus 
that satisfies \textit{(i)}~safety (agreement and validity) provided that fewer than $1/3$ of the players are attackers, 
\textit{(ii)}~liveness (termination) provided that fewer than $2/3$ of the players are attackers and 
\textit{(iii)}~transferable authentication.
Examples of such algorithms include Polygraph~\cite{DBLP:conf/icdcs/CivitGG21}, ABC~\cite{CGGGK23} and ABC++~\cite{CCGG26}.
Due to its transferable authentication property, these protocols enable the identification of guilty players as well as the generation of their undeniable \emph{proof-of-frauds (PoFs)}. Under these circumstances, rational players who want to betray the coalition can pretend to join the coalition and collect contradictory signed messages to gather enough PoFs. At the end of the pre-decision phase, the output 
is an agreement on block $b$ if agreement is reached, or a list of unpublished and encrypted PoFs if the $\zeta$-uple spending attack was successful. For further details, please refer to line~\ref{line:abc-inv} of Alg.~\ref{alg:P} and lines~\ref{line:betray_step_1}--~\ref{line:betray_step_1_end} of Alg.~\ref{alg:P_betray} in Appendix~\ref{app:pseudocode}.
\item {\bf Commit-Reveal phase.}
In the decision phase, all players re-send their encrypted pre-decided block or list of PoFs. The encryption makes it impossible to distinguish betrayers from others or to determine whether a betrayal has occurred before a specified sub-phase; as a result, the Byzantine players cannot pretend to be rational or betrayal players. Because of the two phases, the attack is exposed prior to the final decision. As a result attackers are punished, their stake is slashed and redistributed as rewards to betrayers. Moreover, 
betrayers are not part of the players who run the consensus algorithm
for a predetermined number of rounds, with the protocol slashing a portion of their stake and returning it to them at a later stage. 
Even if a $\zeta$-uple spending attack is successful, 
all PoFs are guaranteed be be eventually received by all honest players who will then resolve forks by choosing one of the blocks in a deterministic fashion. For further details, please refer to lines~\ref{line:start}--~\ref{line:end} of Alg.~\ref{alg:P} in Appendix~\ref{app:pseudocode}.
\end{enumerate}

Flowcharts 2 to 4 in Figure~\ref{fig:fig4} run in parallel with the main function to identify and slash all attackers, define the criteria that a betrayer must meet to be considered valid, and reward these valid betrayers using the penalties collected from attackers. As noted above, betrayers construct a PoF according to Flowchart 5. When a player receives a new PoF, it validates it and transfers the stake of the attacker to the global wallet as a punishment. If the PoF meets the criteria, then the player re-signs it as $\text{Sign}_2$ and broadcasts it, as depicted in Flowchart 2. The corresponding betrayer collects $t_0+1$ of these re-signed PoFs as PoF-ledger and broadcasts it, as depicted in Flowchart 5. 
When a player receives a new PoF-ledger, it validates it. If the PoF-ledger meets the criteria of Flowchart 4, the player gives a reward to the new valid betrayer by deducting a portion of the reward of previous betrayers. The protocol distributes the reward equally among all valid betrayers, since partial synchrony provides no guarantee on message ordering and players cannot reliably determine which betrayer submitted their PoF-ledger first. As soon as the coins enter the global wallet, players transfer and divide the amount among all existing valid betrayers, according to Flowchart 3. The global wallet starts at nil and resets to nil after each transfer to the accounts of the betrayers. For further details, refer to lines~\ref{line:PoF}--~\ref{line:end_parallel} of Alg.~\ref{alg:P} and Alg.~\ref{alg:P_betray} in Appendix~\ref{app:pseudocode}.

\begin{definition}[Attack Attempt]
Let a distributed system consist of $n$ players $\Pi = \{p_1, \dots, p_n\}$, of which attackers $\mathcal{A} \subseteq \Pi$ try to establish $\zeta$-uple spending attack. An \textbf{attack attempt} occurs if 
there exist two sets
\[
Q_1,Q_2 \subseteq \Pi, \qquad Q_1\cap Q_2 \subseteq \mathcal{A} , \qquad |Q_1|,|Q_2| \geq n-t_0 , 
\]
such that every player in $Q_1$ has received a valid message $m_1$ and every player in $Q_2$ has received a valid message $m_2$ at the end of pre-decision phase, in which $m_1$ and $m_2$ are conflicting messages.
\end{definition}

Following the minimum number of rational players given by Theorem~\ref{Theorem:the2}, the following definition names the situation when a betraying coalition happens.
We say that a rational player \emph{betrays} when its chooses to execute Flowchart 5 from Figure~\ref{fig:fig4}.

\begin{definition}[Betraying Coalition]
When an attack attempt occurs, we say that a \emph{betraying coalition} exists if at least $\lfloor \frac{k+t}{2}-\frac{n}{6} \rfloor +1$ rational players betray.
\end{definition}

The following definition is inspired by the Accountable Agreement problem~\cite{DBLP:conf/icdcs/CivitGG21}, except that it requires to output all attackers and betrayers to apply incentives, hence the name.

\begin{definition}[Fully Accountable Agreement]\label{def:FAA}
A protocol solves the 
\emph{Fully Accountable Agreement} problem
if it satisfies all the following properties:
\begin{itemize}
    \item {\bf Detection}: All honest players eventually output all attackers.
    \item {\bf Convergence}: All honest players eventually output the same list of betrayers.
  \item {\bf Agreement}: All honest players end up with the same decision block.
  \item {\bf Validity}: The decided block was proposed by some honest player.
  \item {\bf Termination}: All honest players eventually output a decision block. 
\end{itemize}
\end{definition}

The following theorem states that \protocol solves the fully accountable agreement protocol. The detailed proof is deferred to Appendix~\ref{app:the1}.
This result will be key to show that \protocol implements a blockchain despite $\zeta$-uple attack attempts later. 

\begin{theorem}\label{Theorem:the1}
\protocol 
solves the Fully Accountable Agreement problem for each new index of the blockchain.
\end{theorem}

\section{Tolerating $\zeta$-uple spending attacks}\label{sec:Solution}
In this section, we prove that the \protocol protocol with an appropriate $s/\ell$ gives an $\epsilon$-$(k,t)$-robust equilibrium in the infinitely repeated game $\Gamma$ with $h_0=\{Attack\}$ (this corresponds to the left side of the game tree shown in Appendix~\ref{app:gameTree}, Figure~\ref{fig:fig1}). In particular, Algorithm 1 of the \protocol protocol corresponds to strategy $\rho$ while Algorithm 2 of the \protocol protocol corresponds to strategy $\sigma$ in Figure~\ref{fig:fig4}. Since the betrayal is a dominant strategy, the payoff of mutual betrayal is higher than mutual cooperation and betrayal is self-enforcing even in the infinitely repeated game. In addition, the \protocol protocol slashes betrayers' stake, but returns it to them plus the betraying reward as liquidity. As a result, 
they are not part of the players who run the consensus algorithm
and they cannot initiate another attack instantly. Notably, the Byzantine fault-tolerant consensus properties used in the pre-decision phase are crucial for ensuring that rational players adhere to the protocol when Byzantine players are inactive ($h_0=\{Pass\}$, i.e., the right side of the game tree shown in Appendix~\ref{app:gameTree}, Figure~\ref{fig:fig1}).

\subsection{Why a successful $\zeta$-uple spending attack, with $\zeta>2$, is rare}\label{sec:attack-complexity}

In Section~\ref{sec:problem} we showed that number of blockchain branches $f+k$ attackers can create is large.
In reality, isolating the network for more than two branches is hard due to the complexity of the network attack---this is probably the reason why we are not aware of any successful double spending attack against Bitcoin since 2013~\cite{doublespending}. In addition, the intermediary infrastructure acts as a communication proxy between players and the rest of the blockchain network, hiding their location~\cite{1011453672403}, and if one becomes unresponsive, players reconnect to another that may be in an unpredictable region and to a different Autonomous System (AS)~\cite{kiffer2021under}. Additionally, direct peering between ASes prevents BGP hijacking, making it more difficult to partition networks.
Later, we will explain how every attack attempt produces undeniable proofs-of-fraud, which leads to slashing, so the attacker cannot retry indefinitely with the same stake.
Given that blockchains typically create blocks regularly and that the probability $p$ of isolation attack success is particularly low, 
we assume that
%
the number of isolation attack attempts $m$ per block index multiplied by the probability $p$ of such an attack success is such that $mp < 3/e = 1.1$, which is small but not negligible. This bound will be useful to prove Theorems~\ref{Theorem:the4} and~\ref{Theorem:the5}.
\begin{proposition}\label{prop}
The number of isolation attack attempts $m$ multiplied by the probability $p$ of success is such that $mp < 3/e$.
\end{proposition}

To upper-bound the probability of a $\zeta$-uple spending attack with $\zeta\geq 3$, we conservatively grant attackers up to $m \ge 3$ independent isolation attempts for more than two branches. As a result, we can use Chernoff Bounds~\cite{motwani1995randomized} to show that the probability of a successful $\zeta$-uple spending attack, where $\zeta\geq 3$, is at most $\varepsilon$. 
We defer the proof to Appendix~\ref{app:lemma5}.

\begin{lemma}\label{lem:lemma5}
Suppose attackers are able to launch up to $m$ isolation attacks, each of which succeeds independently with probability $p$. The probability of a successfully $\zeta$-uple spending attack, where $\zeta\geq 3$, is upper bounded by $\varepsilon = \left(e/ \lambda\right)^3$. 
\end{lemma}

\subsection{Attacks without bribery}\label{sec:proofs}
In this section, we show that the \protocol protocol from Section~\ref{sec:solution} is immune to $\zeta$-uple spending attacks in the default case where Byzantine players cannot bribe rational players by offering them money.
Hence, the only incentive for rational players to join the coalition is the personal gain stolen during the attack. We now consider the realistic case where the $\zeta$-uple spending attack, especially when $\zeta\geq 3$, is complex as shown in Section~\ref{sec:attack-complexity}. To this end, we demonstrate that \protocol is a $\epsilon$-$(k,t)$-robust equilibrium.
The proof of the following theorem is deferred to Appendix~\ref{app:the45}.

\begin{theorem}\label{Theorem:the4}
If $\frac{s}{\ell} > \frac{k\cdot p^2}{f+k}$, then \protocol protocol is an $\epsilon$-$(k,t)$-robust equilibrium against $\zeta$-uple spending attacks. 
\end{theorem}

\subsection{Attacks despite bribery}\label{sec:bribery-proofs}

In this section, we show that the \protocol protocol from Section~\ref{sec:solution} is immune to $\zeta$-uple spending attacks even when Byzantine players can use the gain of their attacks to bribe rational players by offering them money. Note that this can be implemented 
on modern blockchains by having Byzantine players who know each other implement a smart contract that redistributes the gain of the attacks exclusively to rational players. In what follows, we list the constraints on the relation between $k$, $f$ and $s$ for the \protocol protocol to be an $\epsilon$-$(k,t)$-robust equilibrium in the realistic scenario where $\zeta$-uple spending attacks are difficult as shown in Section~\ref{sec:attack-complexity}. The 
proof of the following theorem is deferred to Appendix~\ref{app:the45}.

\begin{theorem}\label{Theorem:the5}
If $\frac{s}{\ell} > p^2$, then the \protocol protocol is an $\epsilon$-$(k,t)$-robust equilibrium that protects against $\zeta$-uple spending attack and bribery. 
\end{theorem}

\section{Analysis}\label{sec:analysis}

In this section, we analyze the situations under which the \protocol protocol tolerates the presence of $k$ rational players and $t$ Byzantine players out of which $f$ are detectable attackers. We also show that US\$500 can be a sufficient stake to secure a \protocol protocol blockchain against $\zeta$-uple spending attacks.

\subsection{The Impact of bribery on the required stake}
Figure~\ref{fig:grid} depicts the portions of stake over liquidity that \protocol should require to tolerate $f\leq t$ Byzantine attackers colluding with $k$ rational players in two scenarios where the probability of isolation attack success is $0.1$. The two scenarios correspond to whether Byzantine players can use their profit to bribe rational players.

\begin{figure}[ht!]
    \centering
    \begin{subfigure}{0.45\textwidth}
        \includegraphics[width=\linewidth]{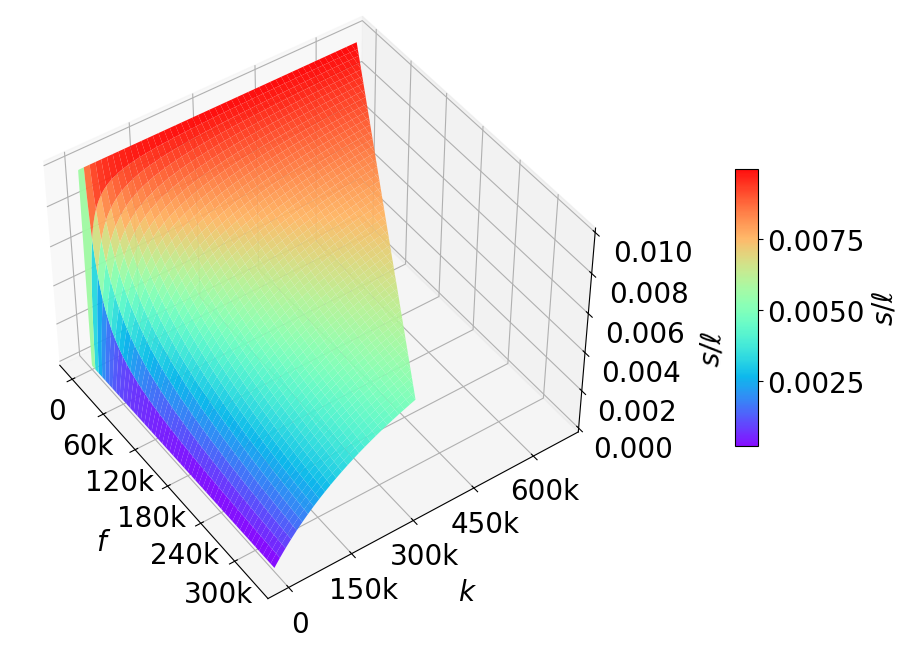}
        \caption{Attack without bribery, $p=0.1$\label{fig:comp_nobib}}
    \end{subfigure}
    \hfill
    \begin{subfigure}{0.45\textwidth}
        \includegraphics[width=\linewidth]{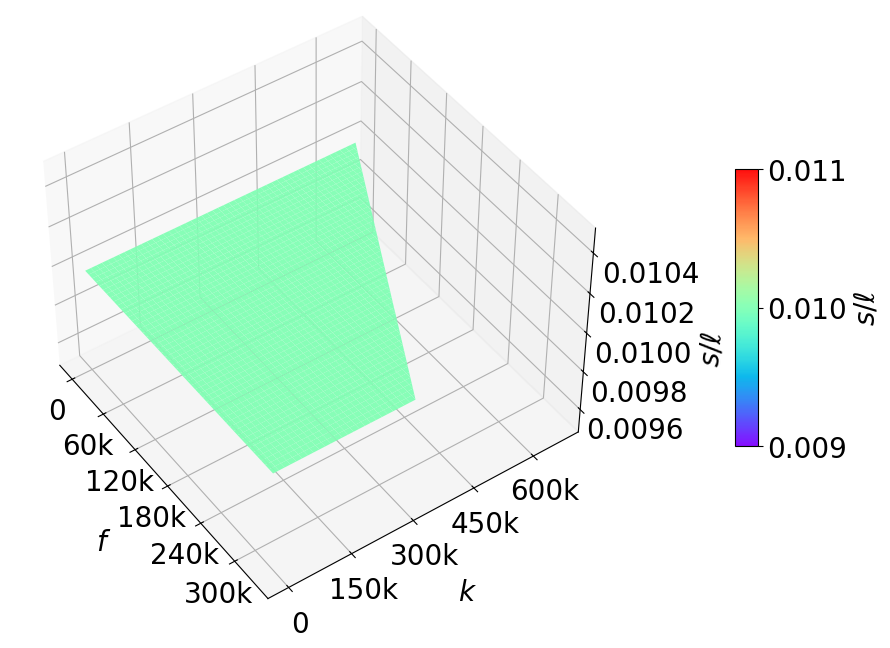}
        \caption{Attack despite bribery, $p=0.1$\label{fig:comp_bib}}
    \end{subfigure}
    \caption{The relation between $k$, $f$ and $\frac{s}{\ell}$ for $n=10^6$. Note that one colored scale is specific to each figure.}
    \label{fig:grid}
\end{figure}

The observation is that 
the bribery plays a significant role. As shown in Figure~\ref{fig:comp_nobib}, if the Byzantine players do not use their gain to bribe rational players to join their coalition, then the required stake over liquidity ratio $s/\ell$ should be sufficient to generate enough reward for betrayers.
As the number of detectable attackers decreases, this ratio should compensate the lack of reward to keep incentivizing rational players to betray.
Finally, Figure~\ref{fig:comp_bib} shows that if Byzantine players can bribe rational players, the possible configurations where the protocol works are more limited but do not require a high stake over liquidity ratio $s/\ell$.
Interestingly, when comparing Figures~\ref{fig:comp_nobib} to~\ref{fig:comp_bib}, one can see that although bribery requires a slightly higher staking over liquidity ratio, its impact on the overall trend is insignificant.

{\renewcommand{\arraystretch}{1.3}
\setlength{\tabcolsep}{10pt}
\begin{table}[ht]
    \centering
    \caption{The minimum ratio of stake $s$
    over liquidity $\ell$ needed to prevent double spending attacks (or $\zeta$-uple attacks with $\zeta \geq 2$), where $p$ is the success probability of an isolation attack,
    $f$ is the number of detected Byzantine players, $h$ is the number of honest players and
    $k = n-f-h$ is the number of rational players.   
    }\label{table:analysis}

    \begin{tabular}{l|cccc}
    \toprule
    & $p$ & $f$ & $h$ & $\min\left\{ \frac{s}{\ell}\right\}$ \\
    \midrule
    Attack without bribery & $0.001$ & $\lceil \frac{n}{3} \rceil - 1$ & $\lfloor \frac{n}{3} \rfloor + 1$ & $5 \cdot 10^{-7}$ \\
         & $0.01$  & $\lceil \frac{n}{3} \rceil - 1$ & $\lfloor \frac{n}{3} \rfloor + 1$ & $5 \cdot 10^{-5}$ \\
         & $0.1$   & $\lceil \frac{n}{3} \rceil - 1$ & $\lfloor \frac{n}{3} \rfloor + 1$ & $5 \cdot 10^{-3}$ \\
         & $0.1$ & $\lceil \frac{n}{3} \rceil - 1$ & $\lfloor \frac{2n}{3} \rfloor$ & $\frac{1}{\lceil n/3 \rceil} \cdot 10^{-2}$ \\
         & $0.1$ & $1$ & $\lfloor \frac{2n}{3} \rfloor$ & $(1-\frac{1}{\lceil n/3 \rceil}) \cdot 10^{-2}$ \\
         \hline
         Attack despite bribery & $0.001$ & $\lceil \frac{n}{3} \rceil - 1$ & $\lfloor \frac{n}{3} \rfloor + 1$ & $10^{-6}$ \\
         & $0.01$  & $\lceil \frac{n}{3} \rceil - 1$ & $\lfloor \frac{n}{3} \rfloor + 1$ & $10^{-4}$ \\
         & $0.1$   & $\lceil \frac{n}{3} \rceil - 1$ & $\lfloor \frac{n}{3} \rfloor + 1$ & $10^{-2}$ \\
         & $0.1$ & $\lceil \frac{n}{3} \rceil - 1$ & $\lfloor \frac{2n}{3} \rfloor$ & $10^{-2}$ \\
         & $0.1$ & $1$ & $\lfloor \frac{2n}{3} \rfloor$ & $10^{-2}$ \\
        \bottomrule
    \end{tabular}
\end{table}
%
}

\subsection{How much stake to secure the \protocol protocol}
Table~\ref{table:analysis} gives some exploitable numerical analysis to secure the \protocol blockchain against $\zeta$-uple spending attacks. First, note that we do not consider the case where $t \geq n/3$ as it is well-known that consensus (and thus blockchain) cannot be solved in an open network~\cite{dwork1988consensus} where the bound on message delays is unpredictable. Therefore, we upper-bound $f\leq t$ to the largest integer lower than $n/3$, which is $\lceil\frac{n}{3}\rceil-1$. In addition, we can afford having more than $n/3$ honest players, therefore we vary the number of honest players between the smallest integer larger than $n/3$, which is $\lfloor\frac{n}{3}\rfloor+1$, and the case where we have 
$\lfloor\frac{2n}{3}\rfloor$. In particular, note that at least one betrayer $k\geq 1$ is necessary, so when $f=\lceil\frac{n}{3}\rceil-1$, we have $h\leq \lfloor\frac{2n}{3}\rfloor$. 

We can see that the ratio of stake needed over liquidity is pretty low.
Specifically, when the probability of success of an isolation attack is particularly low ($p=10^{-3}$), a player needs to stake much less than its liquidity ($s = 5\ell\cdot 10^{-7}$) to secure the network. For example, this is equivalent to requiring someone to stake only US$\$500$ to allow them to transfer US\$1B per transaction. When the probability of success of an isolation attack raises to $0.1$, we can see that a player needs to stake at most $\ell\cdot 10^{-2}$ depending on the number of Byzantine players. For example, it is enough to stake US$\$10$ to allow to transfer US\$1K per transaction, regardless of the bribery situation. As expected, we also observe that the staking over liquidity $s/\ell$ ratio needed to secure the network grows when Byzantine players can bribe rational players.

\subsection{Comparison between bribery and without bribery}

One may wonder why, in the bribery situation, the numerical results differ from those in the situation without bribery. Actually, the answer to this question is that, without bribery, the only variable affecting the profit of rational players is the stake of Byzantine players, which the protocol distributes among betrayers. In a bribery situation, however, as the percentage of Byzantine players increases, they gain greater power to incentivize rational players to join the coalition.


\section{Related Work}\label{sec:rw}

\trap{}~\cite{ranchal2022trap} was the first consensus algorithm to use a baiting strategy to trade honest players for rational players. This is how it reduces the need of honest players from the traditional $\lfloor 2n/3 \rfloor +1$ threshold to $\lfloor n/2 \rfloor +1$, needed for its internal commit-reveal consensus protocol. Unfortunately, it can neither be used in a blockchain where the Folk theorem~\cite{31353} prevents  multi-shot consensus nor reach the $\lfloor n/3 \rfloor + 1$ honest threshold \protocol reaches.

To demonstrate tolerance to rational players, various equilibria have been proposed. Abraham et al.~\cite{DBLP:conf/tcc/AbrahamDH08} introduced the problem of $(k,t)$-robust equilibrium \protocol reaches.
The Perfect Bayesian Equilibria~\cite{YBPT23} copes with Byzantine and rational players but only to 
prevent the proposal of invalid blocks. Nash equilibria were used to determine rewards and the amount of delegated representatives for proof-of-stake governance~\cite{birmpas2025reward}, our goal is however different. 

Ouroboros~\cite{crypto2017:ouroboros} is a proof-of-stake blockchain that reaches a $\delta$-Nash equilibrium to mitigate selfish mining by limiting the gain of a deviating coalition to $\delta$ but requires a majority of the stake to be owned by honest players.
It builds upon Fruitchain~\cite{pass2017fruitchains}
that mitigates the same problem for proof-of-work blockchains
using the $n/2$-coalition-safe $3\delta$-Nash equilibrium where a 
majority of honest players is 
needed for no coalition to gain 
more than $3\delta$.

The Strong Nash Equilibrium (SNE)~\cite{aumarm2016acceptable} prevents a coalition of players to jointly deviate and strictly increase their payoff. However, a consensus protocol cannot be an SNE because there will always exist a coalition large enough to prevent consensus. 
A Coalition-Proof Nash Equilibrium (CPNE)~\cite{coalition1proof1equilibrium} requires all players to be rational, and cannot be used in our case where the utility function of Byzantine players is unknown.
Subgame-perfect pure Nash equilibrium (SPPNE) was used by 
Kelkar et al.~\cite{kelkar2025breaking} who also need all players in the system to be rational. 

Zhang et al.~\cite{Evolutionary1Game2Theory} consider evolutionary game theory where players' behavior is influenced by their neighbourhood while our players' behavior is influenced by their utility function.
%
%
Solidus~\cite{Abraham2016SolidusAI} 
is a blockchain protocol using positive and negative incentives to reach a $(k,t)$-robust equilibrium, however, it requires $k+t<\frac{n}{3}$ while \protocol works as long as $k+t<\frac{2n}{3}$.

ZLB~\cite{10646997} offers a zero-loss payment system for fungible assets as long as there are less than $5n/9$ Byzantine players. Instead of preventing double spending attacks, it resolves forks by reimbursing the victims. Instead \protocol 
protects against double spending attacks.
%

It is known to be impossible for  a state machine replication to reach a ($k,t$)-robust equilibrium in the general model~\cite{avarikioti2025cobra}.
The authors propose CoBRA that bypasses this impossibility by assuming synchrony. CoBRA builds on a previous impossibility result~\cite{Roughgarden} and reaches the same threshold, but it goes further; it ensures the slashed stake is never less than the attack's profit.
To avoid being vulnerable to unexpected delays, \protocol instead requires consensus participants to stake a small portion of their liquidity. 



\section{Discussion and Conclusion}\label{sec:conclusion}

One may wonder whether it is easier to ensure blockchain security in a network with a higher proportion of rational players than Byzantine players, given that Byzantine players target security regardless of the financial incentives. However, our results suggest the opposite. Actually, the answer to this question is that Byzantine players' staked punishment serves as the reward that motivates rational players to betray. Consequently, the more rational players there are, the higher the total reward is and the more Byzantine stake we need to cover the betrayal reward. 

In this paper, we presented the \protocol protocol to secure blockchain against double spending attacks with only $\lfloor n/3 \rfloor + 1$ honest players. To this end, we defined rewards and punishments to incentivize $k$ rational players to behave or betray $t$ Byzantine players.
Our protocol is modeled as an extensive game with imperfect information and perfect recall, and we show that, in an ideal environment, our protocol reaches a $(k, t)$-robust equilibrium with both $t<\frac{n}{3}$ Byzantine and $k=\frac{n}{3}$ rational players.
By remarking that it is difficult to conduct a successful  $\zeta$-uple spending attack, we show that the ratio of stake over liquidity, $s/\ell$, does not need to be high to secure a blockchain.

\subsection*{Acknowledgements}

This work is supported in part by 
the Australian Research Council (ARC) Discovery Project funding
scheme (\#250101739).

\bibliographystyle{splncs04}
\bibliography{references} 

\appendix

\section{Formal Definition of the Extensive Game}\label{app:gameTree}
The set of histories $H$ is given explicitly by
\begin{align*}
H = &\{\emptyset, (Attack), (Pass), (Attack, \neg Bribe), (Attack, Bribe)\} \\
    &\cup \{(Attack, b, a_1, \ldots, a_i) : i \in \{1, \ldots, k\}, b \in \{\neg Bribe, Bribe\}, \\
    &\qquad a_j \in \{Betray, Collude\} \, \forall j \leq i\} \\
    &\cup \{(Pass, a_1, \ldots, a_i) : i \in \{1, \ldots, k\}, a_j \in \{Attack, Commit\} \, \forall j \leq i\}.
\end{align*}

\begin{figure}
\scalebox{0.75}{
\begin{minipage}{\textwidth}
\begin{center}
\small
\begin{tikzpicture}[thin,
  level 1/.style={sibling distance=80mm},
  level 2/.style={sibling distance=40mm},
  level 3/.style={sibling distance=20mm},
  level 4/.style={sibling distance=10mm},
  every circle node/.style={minimum size=1.5mm,inner sep=0mm}]
  
  \node[circle,draw,label=above:$Byzantine$] (root) {}
    child { node [circle,draw,label=above left:$Byzantine$] {}
      child { 
        node[circle,fill,label=above:$R_1$] (node-L1) {}
          child {
            node[circle,fill] (node-L1A) {}
            child {
            node {$u_i(Betray)$}
            edge from parent
              node[left] {$Be$}}
            child {
             node {}
             edge from parent
               node[right] {$Co$}}
            edge from parent
                  node[left] {$Betray$}}
          child {
             node[circle,fill] (node-L1B) {}
            child {
            node {}
            edge from parent
              node[left] {$Be$}}
            child {
             node {$u_i(C)$}
             edge from parent
               node[right] {$Co$}}
             edge from parent
               node[right] {$Collude$}}
          edge from parent
            node[left] {$\neg Bribe$}}
      child { 
        node[circle,fill,label=above:$R_1$] (node-L2) {}
          child {
            node[circle,fill] (node-L2A) {}
            child {
            node {$u_i(B)$}
            edge from parent
              node[left] {$Be$}}
            child {
             node {}
             edge from parent
               node[right] {$Co$}}
            edge from parent
              node[left] {$Betray$}}
          child {
             node[circle,fill] (node-L2B) {}
            child {
            node {}
            edge from parent
              node[left] {$Be$}}
            child {
             node {$u_i(Collude)$}
             edge from parent
               node[right] {$Co$}}
             edge from parent
               node[right] {$Collude$}}
          edge from parent
            node[right] {$Bribe$}}
      edge from parent
        node[left] {$Attack$}}
    child { node [circle,fill,label=above:$R_1$] {}
      child { 
        node[circle,fill] (node-R1) {}
          child {
            node {$u_i(Attack)$}
            edge from parent
              node[left] {$Attack$}}
          child {
             node {}
             edge from parent
               node[right] {$Commit$}}
          edge from parent
            node[left] {$Attack$}}
      child { 
        node[circle,fill] (node-R2) {}
          child {
            node {}
            edge from parent
              node[left] {$Attack$}}
          child {
             node {$u_i(Commit)$}
             edge from parent
               node[right] {$Commit$}}
          edge from parent
            node[right] {$Commit$}}
       edge from parent
         node[right] {$Pass$}};
  
  \draw[dashed] (node-L1A) -- (node-L1B) 
     node[midway,above] {$R_2$};
  \draw[dashed] (node-L2A) -- (node-L2B)
     node[midway,above] {$R_2$};
  \draw[dashed] (node-R1) -- (node-R2)
     node[midway,above] {$R_2$};
  
\end{tikzpicture}
\end{center}
\end{minipage}
}
\caption{Extensive-form game with imperfect information and perfect recall for two rational players $R_1$ and $R_2$. $Be$ is an abbreviation of $Betray$ and $Co$ is an abbreviation of $Collude$.} \label{fig:fig1}
\end{figure}
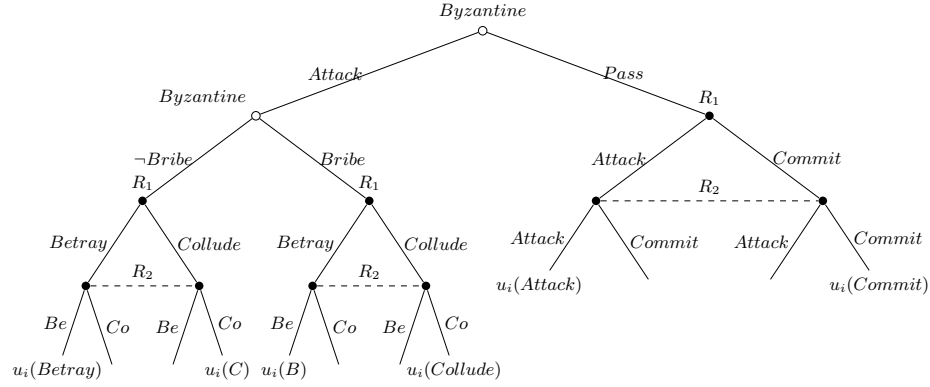

Figure~\ref{fig:fig1} shows the game tree of $\Gamma$ for two rational players (for each additional rational player, a level will be added to the tree vertically). Byzantine players form a coalition and initiate an attack by controlling the scheduler in the first global round. 
In each global round, if rational players have sufficient power to initiate an attack, they do so, and this corresponds to the left side of the tree. 
The system security defines the probability distribution of Nature's move. 
We exclude honest players from the game tree as they always follow the protocol and only their number affects the utility function of rational players. 
Additionally, in the worst case scenario we focus on here, they make a coalition and we can consider them as a single player of the game. After the move of Nature or the scheduler, all rational players make their moves simultaneously under \emph{endogenous uncertainty}, meaning they are uncertain about another player's choice.

\section{Pseudocode for the \protocol protocol}\label{app:pseudocode}
\begin{algorithm}[p]
  \caption{\protocol protocol for player $i$} 
  \label{alg:P}
   \scriptsize{
    \begin{algorithmic}[1]
        \Part{{\bf State}}{
            \State $\ms{t_0}=\lceil \frac{n}{3} \rceil-1$
            \State $\ms{s_i}$, the stake of player $i$
            \State $\ms{\ell_i}$, the liquidity of player $i$
            \State $\ms{blockchain}$, the list of blocks, initially the genesis block
            \State $j$, index of the last block, initially 0
            \State $r$, block index
            \State $\mathit{Attacker} := (\mathit{id},\ s,\ \ell)$
            \State $\mathit{Betrayer} := (\mathit{id},\ s,\ \ell)$
            \State $\ms{pof}$, the structure of proof of fraud, $\ms{pof} := (\mathit{evidence},\ \mathit{attackers},\ \mathit{betrayer},\ r)$ where $\mathit{attackers} \subseteq \mathit{Attacker}$, $\mathit{betrayer} \in \mathit{Betrayer}$
            \State $\ms{pofs}$, the dictionary of all proofs of fraud, $\ms{pofs} : r \mapsto (\mathit{attackers}, \mathit{betrayers})$ where $\mathit{attackers} \subseteq \mathit{Attacker}$, $\mathit{betrayers} \subseteq \mathit{Betrayer}$, initially $\varnothing$
            \State $\ms{pof-ledger}$, the dictionary of second signatures of a single $\ms{pof}$, $\ms{pof-ledger} : signer \mapsto \ms{pof}$
            \State $\ms{key}$, the encryption private key of player for commit-reveal
            \State $x$, index of other players
            \State $\ms{list-enc-msgs}$, list of delivered encrypted messages from commit-reveal broadcasts, initially $\varnothing$
            \State $\ms{keys}$, list of delivered encryption private keys from the third group broadcasts, initially $\varnothing$
            \State $\ms{gwallet}$, global wallet, list of total slash stakes before and after allocating to betrayers, initially $\varnothing$
            \State $\ms{pof-batch}$, the set of proofs of fraud received during the commit phase, initially $\varnothing$
            \State $\mathrm{dom}(\cdot)$ denotes the domain (set of keys) of a map
        }\EndPart

        \Statex \rule{0.45\textwidth}{0.4pt}
        \Part{\textbf{Upon creation of new block $b$}}{
            \State ~~~~~$\ms{blockchain}[j] = \Call{append-block}{\ms{b}}$ \label{line:new-block}
            \State ~~~~~$j= j+1$
        }\EndPart
        
        \Statex

        \Part{\textbf{Upon receiving a new valid $\ms{pof}$}}{ \Comment{identical valid signature for conflict transactions}\label{line:PoF}
            \State ~~~~~$\ms{pofs}[r].\mathit{attackers} \leftarrow \ms{pofs}[r].\mathit{attackers} \cup \ms{pof}.\mathit{attackers}$
            \State ~~~~~$\ms{gwallet}[before][r] \mathrel{+}= \sum_{a \,\in\, \ms{pofs}[r].\mathit{attackers}} a.s$ \label{line:StoGlobalWallet}
            \State ~~~~~$\forall\, a \in \ms{pofs}[r].\mathit{attackers}: \; a.s = 0$ \Comment{slash attackers stake}\label{line:CoC}
            \State ~~~~~\textbf{if} $\ms{pof} \in \ms{pof-batch}$ \textbf{then}
            \State ~~~~~~~~~~$\ms{double-signed-pof}=\lit{sign}(\ms{pof})$\label{line:second_signature}
            \State ~~~~~~~~~~$\lit{broadcast}(\ms{double-signed-pof})$\label{line:second_signature_bcast}
            \State ~~~~~\textbf{else}
         $\lit{broadcast}(\ms{pof})$ \label{line:no_second_signature_bcast}
        }\EndPart

        \Statex

        \Part{\textbf{Upon receiving a new valid $\ms{pof-ledger}$}}{ \label{line:PoFledger}
            \State ~~~~~Let $b \in \mathrm{dom}(\ms{pof-ledger})$ be arbitrary.
            \State ~~~~~\textbf{if} 
            $\forall\, b_1, b_2 \in \mathrm{dom}(\ms{pof-ledger}): \ms{pof-ledger}[b_1] = \ms{pof-ledger}[b_2] \; \land$\label{line:con5} \Comment{contain identical $\ms{pof}$}
            \State ~~~~~$|\ms{pof-ledger}[b].\mathit{attackers}| = \ms{t_0}+1 \; \land$ \label{line:con1} \Comment{reveal $\ms{t_0}+1$ attackers}
            \State ~~~~~$|\ms{pof-ledger}| \geq \ms{t_0}+1 \; \land$ \label{line:con2} \Comment{contain $\ms{t_0}+1$ second signatures}
            \State ~~~~~$\mathrm{dom}(\ms{pof-ledger}) \cap \{a : a \in \ms{pof-ledger}[b].\mathit{attackers}\} = \varnothing \; \land$ \label{line:con3} \Comment{
            2nd signers are not  attackers
            }
            \State ~~~~~$\ms{pof-ledger}[b].\mathit{betrayer} \in \ms{pof-ledger}[b].\mathit{attackers}$ \label{line:con4} 
            \textbf{then} \Comment{reveal itself}
            \State ~~~~~~~~~~$r'=\ms{pof-ledger}[b].r$
            \State ~~~~~~~~~~$sz=\lvert \ms{pofs}[r'].\mathit{betrayers}\rvert$
            \State ~~~~~~~~~~$\mathit{new\_btr}=\ms{pof-ledger}[b].\mathit{betrayer}$
            \State ~~~~~~~~~~$\forall\, \beta \in \ms{pofs}[r'].\mathit{betrayers}: \; \beta.\ell \mathrel{-}= \frac{\lit{\ms{gwallet}[after][r']}}{(sz+1) \cdot sz}$ \Comment{
            take reward from prev. betrayers}
            \State ~~~~~~~~~~$\ms{pofs}[r'].\mathit{betrayers} \leftarrow \ms{pofs}[r'].\mathit{betrayers} \cup \{\mathit{new\_btr}\}$   
            \State ~~~~~~~~~~$\ms{pofs}[r'].\mathit{new\_btr}.\ell \mathrel{+}= \frac{\lit{\ms{gwallet}[after][r']}}{sz}$ \Comment{and paid to the new betrayer}\label{line:BoR2}
            \State ~~~~~\textbf{end if}
        }\EndPart
        
        \Statex
        
        \Part{\textbf{Upon $\ms{gwallet}[before][r]\neq \varnothing$}}{ \label{line:GlobalWallet}
            \State ~~~~~\textbf{if} $\ms{pofs}[r].\mathit{betrayers}\neq \varnothing$ \textbf{then}
            \State ~~~~~~~~~~$\forall\, \beta \in \ms{pofs}[r].\mathit{betrayers}: \; \beta.\ell \mathrel{+}= \frac{\lit{\ms{gwallet}[before][r]}}{\lvert \ms{pofs}[r].\mathit{betrayers}\rvert}$  \Comment{reward betrayers}\label{line:BoR}
            \State ~~~~~~~~~~$\ms{gwallet}[after][r]\mathrel{+}=\ms{gwallet}[before][r]$
            \State ~~~~~~~~~~$\ms{gwallet}[before][r]=0$
            \State ~~~~~\textbf{end if}
        }\EndPart \label{line:end_parallel}
          
        \algstore{alg:bating}
        
    \end{algorithmic}
   } 
\end{algorithm}

\begin{algorithm}[tp]
  \ContinuedFloat
  \caption{\protocol protocol for player $i$ (continued)} 
   \scriptsize{
    \begin{algorithmic}[1]

        \algrestore{alg:bating}

        \Statex

        \Function{append-block}{$\ms{b}$} \label{line:append-block}
            \State $(\ms{pre-decision} \mid \ms{pof})=\lit{byzantine-agreement(\ms{b})}$ \Comment{agree or identify guilty players} \label{line:abc-inv}
            \State $\ms{block} = \Call{commit-reveal}{\ms{pre-decision} \mid \ms{pof}, \ms{key}}$ \Comment{reveal decision} \label{line:commit-reveal-inv}
            \State \Return $\ms{block}$
        \EndFunction

        \Statex

        \Function{commit-reveal}{$\ms{block} \mid \ms{pof},\ms{key}$}\label{line:start}
            \State $\lit{encrypt}(\ms{block} \mid \ms{pof},\ms{key})$ \Comment{commit phase} \label{line:encrypt}
            \State $\lit{broadcast}(\ms{enc\_msg})$  \label{line:bcast}
            \While{$\lit{size(\ms{list-enc-msgs})}<n-\ms{t_0}$} \label{line:while}
                \State $\ms{list-enc-msgs}[x] \gets \ms{enc\_msg}_x$ \Comment{add to $\ms{list-enc-msgs}$ upon new $\lit{delivery}$}
            \EndWhile
            \State $\lit{broadcast}(\ms{list-enc-msgs})$  \label{line:bcast2}
            \While{$\lit{size(\ms{list-2-enc-msgs})}<n-\ms{t_0}$}
                \State $\ms{list-2-enc-msgs}[x] \gets \ms{list-enc-msgs}_x$ \Comment{add to $\ms{list-2-enc-msgs}$ upon new $\lit{delivery}$}
            \EndWhile
            \State $\lit{broadcast}(\ms{key})$ \Comment{reveal phase} \label{line:reveal} 
            \While{$\lit{receiving all required keys}$}
                \State $\ms{keys}[x] \gets \ms{key}_x$ \Comment{add to $\ms{keys}$ upon new $\lit{delivery}$}
            \EndWhile
            \State $(\ms{block} \mid \ms{pof-batch}) = \lit{decrypt} (\ms{list-enc-msgs},\ms{keys})$ 
            \If{$\ms{pof-batch} \neq \varnothing$}
                \State \Return $\lit{select-block} 
                (\ms{pof-batch})$ \label{line:resolve}
                \Comment{resolve forks}
                \State $\ms{pof-batch} \gets \varnothing$

            \Else
                \State \Return $\ms{block}$  \label{line:single-block} \Comment{return decision}
            \EndIf
        \EndFunction \label{line:end}

    \end{algorithmic}
   } 
\end{algorithm}

\begin{algorithm}[t]
  \caption{\protocol protocol for betrayer $i$} 
  \label{alg:P_betray}
   \scriptsize{
    \begin{algorithmic}[1]
        \Part{{\bf State}}{
            \State $r$, block index
            \State $\mathit{Attacker} := (\mathit{id},\ s,\ \ell)$
            \State $\mathit{Betrayer} := (\mathit{id},\ s,\ \ell)$
            \State $\mathit{Evidence} := (\mathit{id},\ signed\_message_1,\ signed\_message_2)$
            \State $\ms{pof}_i$, the structure of proof of fraud created by betrayer $i$, $\ms{pof}_i := (\mathit{evidence},\ \mathit{attackers},\ \mathit{betrayer},\ r)$ where $\mathit{attackers} \subseteq \mathit{Attacker}$, $\mathit{betrayer} \in \mathit{Betrayer}$
            \State $\ms{pof-ledger}_i$, list of double signed of the $\ms{pof}_i$, initially $\varnothing$
        }\EndPart

        \Statex \rule{0.45\textwidth}{0.4pt}
        \Part{\textbf{Upon joining a coalition in the $\mathit{pre\text{-}decision}$ phase (line~\ref{line:abc-inv}, Algorithm~\ref{alg:P})}}{\label{line:betray_step_1}
            \State ~~~~~$\ms{pof}_i.\mathit{betrayer} \leftarrow \mathit{Betrayer}_i$
            \State ~~~~~$\ms{pof}_i.\mathit{attackers} \leftarrow \ms{pof}_i.\mathit{attackers} \cup \{\mathit{Betrayer}_i\}$ \Comment{reveal itself} \label{line:betray_itself}
            \State ~~~~~$\mathit{pof_i.evidence} \gets \mathit{pof_i.evidence} \cup \{\mathit{Evidence}_i\}$ \Comment{its own conflicting messages}
            \State ~~~~~$\ms{pof}_i.\mathit{r} \leftarrow r$
            \State ~~~~~\textbf{while} {$|\ms{pof}_i.\mathit{attackers}|<t_0+1$}\textbf{ do} \Comment{reveal other attackers}
                \State ~~~~~~~~~~$\ms{pof}_i.\mathit{attackers} \leftarrow \ms{pof}_i.\mathit{attackers} \cup \{\mathit{Attacker}\}$ 
                \State ~~~~~~~~~~$\mathit{pof_i.evidence} \gets \mathit{pof_i.evidence} \cup \{e \in \mathit{Evidence} : e.id = \mathit{attacker}.id\}$
            \State ~~~~~\textbf{end while} \Comment{$\ms{pof}_i$ is attached to the output of $\lit{byzantine\text{-}agreement}$}
        }\EndPart \label{line:betray_step_1_end}
        
        \Statex

        \Part{\textbf{Upon receiving a $\ms{double-signed-pof}_i$ (line~\ref{line:second_signature}, Algorithm~\ref{alg:P})}}{\label{line:betray_step_2}
            \State ~~~~~$r'_1=\ms{double-signed-pof}_i.r$
            \State ~~~~~$\forall\, b \in \mathrm{dom}(\ms{pof-ledger}_i): r'_2=\ms{pof-ledger_i}[b].r$
            
            \State ~~~~~\textbf{if } $\{\ms{double-signed-pof}_i.signer \} \cap \{a : a \in \ms{pof-ledger}_i[b].\mathit{attackers}\} = \varnothing \; \land$ \Comment{no overlap...}
            \State ~~~~~ $r'_1=r'_2$
             \textbf{ then} \Comment{...between second signatories and attackers}
            \State ~~~~~~~~~~$\ms{pof-ledger}_i \leftarrow \ms{pof-ledger}_i \cup \ms{double-signed-pof}_i$
            \State ~~~~~\textbf{end if}
            \State ~~~~~\textbf{if }{$|\ms{pof-ledger}_i| = t_0+1$}\textbf{ then}
                \State ~~~~~~~~~~$\lit{broadcast}(\ms{pof-ledger}_i)$
            \State ~~~~~\textbf{end if}
        }\EndPart

    \end{algorithmic}
   } 
\end{algorithm}

The pseudocode of the \protocol protocol is depicted in Algorithms~\ref{alg:P} and~\ref{alg:P_betray}.
To implement a blockchain (that needs multi-shot consensus), \protocol repeats a single-shot consensus protocol  similar to the \trap{} consensus protocol~\cite{ranchal2022trap} in that it baits rational players to betray a coalition.
Each time a player $i$ creates a new block $b$ it tries to append it to the blockchain by invoking function $\Call{append-block}{b}$ at line~\ref{line:new-block}.
This function consists of invoking 
an accountable Byzantine agreement (line~\ref{line:abc-inv}) that either returns a unique block $\ms{pre-decision}$ or some proofs-of-fraud $\ms{pofs}$ attesting of guilty players attempting a $\zeta$-uple spending attack. A $\Call{commit-reveal}$ function is then invoked with this returned block as a parameter at line~\ref{line:commit-reveal-inv}. Each player exchanges the $\ms{pre-decision}$ or $\ms{pof}$ 
through broadcast in the $\Call{commit-reveal}$
function at line~\ref{line:bcast} before revealing them
at line~\ref{line:reveal}. This process ensures that once a block is pre-decided, it cannot be changed and is kept hidden until the reveal step at the end.
Provided that the $\ms{pof}$s are valid undeniable proofs-of-fraud, the stake of attackers is slashed (line~\ref{line:CoC}) and used to reward betrayers (line~\ref{line:BoR}) who revealed the $\ms{pof}$s. At the end, either a fork involving conflicting blocks induced by $\ms{pof}$s needs to be resolved (line~\ref{line:resolve}) or the unique block $\ms{pre-decide}$ is returned (line~\ref{line:single-block}).

\section{Proof of Theorem~\ref{Theorem:the2}}\label{app:the2}
\begin{lemma}\label{lem:lemma1}
If the number $t+k$ of attackers is strictly larger than the specific threshold 
$\frac{\frac{2}{3}\cdot \zeta-1}{\zeta-1} n$, 
i.e,  
$t+k>\frac{\frac{2}{3}\cdot \zeta-1}{\zeta-1} n$,
then they can establish a $\zeta$-uple spending attack where $\zeta \in \mathbb{Z}^+$.
\end{lemma}

\begin{proof}
In a $\zeta$-uple spending attack, attackers should create $\zeta$ disjoint groups of honest players to be able to create $\zeta$ branches and spend their money $\zeta$ times:
\begin{gather*}
\begin{cases}
h_1+(t+k)>\frac{2n}{3}\\
\vdots\\
h_\zeta+(t+k)>\frac{2n}{3}\\
h_1+\cdots +h_\zeta=h=n-(t+k)\\
\end{cases}
\Rightarrow n-(t+k) + \zeta \cdot (t+k)>\zeta\times \frac{2n}{3}\\
\Rightarrow t+k>\frac{\frac{2}{3}\cdot \zeta-1}{\zeta-1} n.
\end{gather*}
And the proof follows. \qed
\end{proof}

\begin{theorem}[Theorem~\ref{Theorem:the2}]
If $t + k$ attackers attempt to carry out a $\zeta$-uple spending attack, and $\lfloor \frac{k+t}{2}-\frac{n}{6} \rfloor +1$ rational players are randomly selected from a set of $k$ rational players and they send message $m'$, then all honest players receive $m'$.
\end{theorem} 

\begin{proof}
For a successful $\zeta$-uple spending attack, each branch should have more than $\frac{2n}{3}$ players. The worst case for attackers is that they have more than $t_0+1$ players but create only two branches. This means that the remaining $t+k-t_0-1$ attackers are redundant and the attack succeeds even if the messages from these redundant attackers reach one of the two branches. Now suppose that one of these redundant attackers betrays (sends $m'$); it may be unable to notify honest players of both branches before they terminate. We use the pigeonhole principle to find the minimum number of betrayers we need to notify all honest players before the attack terminates. We assume $k=2k_1+k_2$ and partition the k attackers into three disjoint groups. Before termination, both branches receive messages from $k_2$, branch one receives messages from half of $2k_1$ and the other branch receives messages from the other half of $2k_1$. As a result, if $k_1+1$ attackers send message $m'$, all honest players receive that before termination.
\begin{gather*}
\begin{cases}
\frac{h}{2}+k_1+t+k_2>\frac{2n}{3}\\
k=2k_1+k_2\\
n=h+t+k
\end{cases}
\Rightarrow k_1<\frac{k+t}{2}-\frac{n}{6}.
\end{gather*}
According to the pigeonhole principle, by selecting $\lfloor \frac{k+t}{2}-\frac{n}{6} \rfloor +1$ rational players randomly and having them send message $m'$, then all honest players receive at least one $m'$. \qed
\end{proof}

\section{Proof of Theorem~\ref{Theorem:the1}}\label{app:the1}
\begin{lemma}\label{lem:1betrayer}
If we are in the presence of a betraying coalition, then at least one betrayer will be identified eventually as a valid betrayer by all honest players in Algorithm~\ref{alg:P}.
\end{lemma}
\begin{proof}
According to the condition of a valid betrayer from line~\ref{line:con5} to line~\ref{line:con4} in Algorithm~\ref{alg:P} (which is satisfied by line~\ref{line:betray_step_1} in Algorithm~\ref{alg:P_betray}), each betrayer collects $t_0+1$ conflicting messages, including its own conflicting message, which happens since for a successful $\zeta$-uple spending attack, attackers broadcast at least $t_0+1$ conflicting messages in the network in the pre-decision phase. After revealing the $\ms{key}$ at line~\ref{line:reveal} and decrypting $\ms{list-2-enc-msgs}$, upon subsequently noticing a valid $\ms{pof}$ in the list of delivered messages, each honest player will sign the $PoF$ again and broadcast $\ms{double-signed-pof}$, as specified at lines~\ref{line:second_signature} and~\ref{line:second_signature_bcast}. The betrayer $b_i$ should collect at least $t_0+1$ second signatures of its own $\ms{pof_i}$ as $\ms{pof-ledger_i}$ and broadcast it to identify as a valid betrayer and win the reward (explained in Algorithm~\ref{alg:P_betray}, line~\ref{line:betray_step_2}). To be sure that there exists at least one betrayer $b_i$ who is recognized as valid, there should be $t_0+1$ second signatures of $\ms{pof_i}$, which means that at least $t_0+1$ honest players should observe $\ms{pof_i}$ among the first $n-t_0$ messages received during the broadcasts at lines~\ref{line:bcast} and~\ref{line:bcast2}. We refer to this condition as observing $\ms{pof_i}$ at the right time, since only then will $\ms{pof_i}$ be saved in $\ms{list-2-enc-msgs}$ and eventually allow $b_i$ to be identified as a valid betrayer in the blockchain. The reason we insist on observing by honest players is that the owners of these second signatures could not be in the attacker's $\ms{pof_i}$ (line~\ref{line:con3}). Moreover, it is almost impossible for a Byzantine player to sign $\ms{pof_i}$ and broadcast it, which helps identify the Byzantine coalition as the attacker. Additionally, other rational players will not sign and broadcast $\ms{pof_i}$ because if there are more valid betrayers, each rational player will receive a smaller share of the reward, and their profit will decline. 

Now, we need to prove that one common $\ms{pof}$ is inside $\ms{list-2-enc-msgs}$ of at least $t_0+1$ honest players. The proof proceeds by contradiction and the pigeonhole principle. Suppose, for the sake of contradiction, that the statement is false. We then construct the following counterexample. Consider the scenario in which $k_x$ rational players betray and are seen by at most $t_0$ honest players; i.e., $k_x$ only have connections with $\{h_1,\dots,h_{t_0}\}$. In this proof, having a \textit{connection} means the message was received among the first $n-t_0$ total messages by the receiver. Moreover, there is no connection that involves exactly one intermediary between $k_x$ and $h_x \notin \{h_1,\dots,h_{t_0}\}$; because of the second broadcast at line~\ref{line:bcast2}, if there exists one intermediary connection between $k_x$ and $h_x$, $h_x$ will be aware of the corresponding $\ms{pof}$ before revealing keys:
\begin{equation*}
    |\{h_1,\dots,h_{t_0}\}|+t+k_x=t_0+t+k_x=n-t_0 \Rightarrow k_x=n-t-2t_0.
\end{equation*}
Excluding these players, we still have $h-t_0$ honest players who cannot have a connection with $\{h_1,\dots,h_{t_0}\}$ and $k_x$ rational players; otherwise, these $h-t_0$ honest players will be aware of $k_x$, which contradicts the initial statement that at most $t_0$ honest players will see a common $\ms{pof}$. It is worth mentioning that we assume Byzantine players behave arbitrarily, but rational players persist in broadcasting all encrypted messages even if they are not betraying, which means they reveal encrypted $\ms{pof}$s. As a result:
\begin{equation*}
    (h-t_0) + (k-k_x)+t \geq n-t_0 \Rightarrow 2t_0 +t\geq n  \Rightarrow \bot.
\end{equation*}
By the pigeonhole principle, at least $t_0+1$ honest players will see a common $\ms{pof}$, which means there exists at least one valid betrayer eventually.

It is worth mentioning that according to Theorem~\ref{Theorem:the2}, if the number of betrayers is lower than $\lfloor \frac{k+t}{2}-\frac{n}{6} \rfloor +1$, it may happen that their $\ms{pof}$s are not seen by enough honest players at the right time. \qed 
\end{proof}

For the next lemma, we refer to the set of honest players as $\Pi_h$.

\begin{lemma}\label{lem:identify_attackers_betrayers}
Consider we are in the presence of a betraying coalition. All honest players will agree eventually on the same list of valid betrayers and all attackers.
More precisely, $\forall p_i,p_j \in \Pi_h,$ starting with lists of valid betrayers  $\mathcal{S}_i^{\mathcal{Z}_0}$, $\mathcal{S}_j^{\mathcal{Z}_0}$ and lists of detectable attackers $\mathcal{W}_i^{\mathcal{Z}_0}$, $\mathcal{W}_j^{\mathcal{Z}_0}$ at block index $\mathcal{Z}_0$, eventually $\mathcal{S}_i^{\mathcal{Z}_q} = \mathcal{S}_j^{\mathcal{Z}_q}$ for all $q \ge m$ and $\mathcal{W}_i^{\mathcal{Z}_{q'}} = \mathcal{W}_j^{\mathcal{Z}_{q'}}=\mathcal{V}$ for all $q' \ge m'$.
\end{lemma}
\begin{proof}
By examination of Algorithm~\ref{alg:P}, we prove that at least $t_0+1$ attackers will eventually be identified. This is the process of betrayal: the betrayer $b_i$ collects at least $t_0+1$ conflicting messages as $\ms{pof_i}$ (returned by $\lit{byzantine-agreement}$ at line~\ref{line:abc-inv}, explained at line~\ref{line:betray_step_1} in Algorithm~\ref{alg:P_betray}). This is a condition for being recognized as a valid betrayer at line~\ref{line:con1}, and each betrayer can eventually collect $t_0+1$ conflicting messages, since a successful attack requires at least $t_0+1$ conflicting messages. $\ms{pof_i}$ must have proof against the owner as well, according to line~\ref{line:con4} (and line~\ref{line:betray_itself} in Algorithm~\ref{alg:P_betray}), since it stops the owner of rational players from censoring $\ms{pof}$s that contain their ID as the attackers to prevent slashing. As soon as the betrayer $b_i$ collects $t_0+1$ conflicting messages, it will encrypt (line~\ref{line:encrypt}) and broadcast (line~\ref{line:bcast}) them as $\ms{pof_i}$. Even though the betrayer $b_i$ can collect more than $t_0+1$ conflicting messages to yield more rewards, it will not. Because to be identified as a valid betrayer, $\ms{pof_i}$ should be delivered to honest players through the first $n-t_0$ delivered messages to be saved in their $\ms{list-2-enc-msgs}$; so the betrayer $b_i$ is unwilling to take the risk of waiting any longer. It is worth mentioning that the betrayer $b_i$ rushing to broadcast $\ms{pof_i}$ does not mean it will send $\ms{pof_i}$ without encryption. We prove this by contradiction. Sending $\ms{pof_i}$ without encryption will prematurely notify Byzantine players that betrayal has happened, so Byzantine players will start sending the $\ms{pof}$ as well to prevent being slashed or prevent $\ms{pof_i}$ from being propagated. In this situation, the betrayer $b_i$ has an expected utility of zero, so premature sending yields nothing, and $b_i$ would never choose to act that way.

So far, we have proven that at least $t_0+1$ attackers will eventually be identified and punished by all honest players. But what if there are more attackers? Rational players who are detected as valid betrayers (according to Lemma~\ref{lem:1betrayer}, we know that at least one valid betrayer exists) will collect this extra $\ms{pof}$ and will broadcast it to increase their reward (the source of betrayers' reward is the punishment of attackers). Eventually, all honest players will receive this extra $\ms{pof}$ and punish newly identified attackers at line~\ref{line:PoF}. As soon as the $\ms{gwallet}$ is non-empty at line~\ref{line:GlobalWallet}, these extra coins will be distributed among all valid betrayers for the corresponding block index. As a result, it is profitable for rational players to try to gather complete information about attackers over time. Consequently, all honest players will eventually be aware of all attackers, and they will have the same list of attackers. It is worth mentioning that for the received $\ms{pof}$ that is valid but is not inside $\ms{list-2-enc-msgs}$, players will punish attackers and broadcast unsigned $\ms{pof}$ at line~\ref{line:no_second_signature_bcast}. As a result, eventually all players will be aware of these attackers.

Until now we proved that all honest players will eventually detect all attackers. We now demonstrate that all honest players observe an identical list of valid betrayers. At line~\ref{line:PoF}, the moment a player receives a $\ms{pof}$, it will broadcast this $\ms{pof}$ to ensure every honest player on the network observes it. If that $\ms{pof}$ is delivered inside the first and second broadcasts of the commit-reveal phase (lines~\ref{line:bcast} and~\ref{line:bcast2} and be saved in $\ms{list-2-enc-msgs}$), the player should sign the $\ms{pof}$ as a $\ms{double-signed-pof}_i$ at line~\ref{line:second_signature} before broadcasting. This second signature stops each Byzantine player from pretending to be a betrayer when it learns of the betrayal. A betrayer $b_i$ should collect at least $t_0+1$ of these $\ms{double-signed-pof}_i$ of its own $\ms{pof_i}$ (a condition of being a valid betrayer at line~\ref{line:con2}, which is built according to line~\ref{line:second_signature} in Algorithm~\ref{alg:P_betray}) as the $\ms{pof-ledger}_i$. The second signatories should not be inside the relevant attackers' $\ms{pof}_i$, as shown at line~\ref{line:con3}. According to Lemma~\ref{lem:1betrayer}, all honest players will recognize at least one valid betrayer $b_i$ and sign its $\ms{pof_i}$; as a result, collecting $t_0+1$ amount of $\ms{double-signed-pof}_i$s will happen. The reason for the $t_0+1$ amount of $\ms{double-signed-pof}$ is that some silent Byzantine players may exist in the blockchain, and if a betrayal occurs, these silent Byzantine players begin double-signing the late $\ms{pof}$s of other Byzantine players to prevent punishment. Due to the $t<t_0+1$ threshold, at least one of these $\ms{double-signed-pof}$s should be signed by an honest player, which is impossible because an honest player would not sign a $\ms{pof}$ that is published outside the eligible period. Due to a partial synchrony network, all honest players will eventually receive the entire eligible $\ms{pof-ledger}$ at line~\ref{line:PoFledger}, which means they observe the same list of valid betrayers. \qed 
\end{proof}

\begin{definition}
Algorithm~\ref{alg:P} \emph{terminates} if all honest and rational players have reached line~\ref{line:end}.
\end{definition}

\begin{lemma}\label{lem:termination}
Algorithm~\ref{alg:P} will terminate eventually. 
\end{lemma}
\begin{proof}
The consensus algorithm executed during the pre-decision phase intrinsically guarantees termination. Therefore, a player cannot get stuck at line~\ref{line:abc-inv}. Moreover, it cannot get stuck in the commit-reveal function at line~\ref{line:commit-reveal-inv}. We assume Byzantine players behave arbitrarily, but rational players persist in broadcasting all messages even if they are not betraying. Since the number of Byzantine players is less than $t_0+1$, a player eventually receives $n-t_0$ messages from the broadcasts at lines~\ref{line:bcast} and ~\ref{line:bcast2}; consequently, the while condition will eventually be satisfied. A similar argument holds for the key broadcast at line~\ref{line:reveal}. As a result, the append-block function at line~\ref{line:append-block} will eventually terminate. The event-based lines~\ref{line:PoF},~\ref{line:PoFledger} and~\ref{line:GlobalWallet} execute concurrently with the append-block function and do not disrupt its execution. If a new message meets the event-based lines' conditions, the player assigns tags to specific players and changes the balance of the accounts, so the player terminates immediately. For the reasons above, we conclude that players continue executing increasing block indexes; i.e., no process gets stuck at a particular block index. \qed 
\end{proof}

\begin{theorem}[Theorem~\ref{Theorem:the1}]
\protocol solves the Fully Accountable Agreement problem for each new index of the blockchain.
\end{theorem}

\begin{proof}
By Definition~\ref{def:FAA}, if \protocol satisfies properties detection, convergence, agreement, validity and termination then \protocol solves the Fully Accountable Agreement problem. Lemma~\ref{lem:identify_attackers_betrayers} shows that properties detection and convergence are satisfied by \protocol and Lemma~\ref{lem:termination} shows that property termination is satisfied by \protocol. Moreover, the \protocol protocol outputs agreement on a single block by line~\ref{line:resolve}, which satisfies the agreement property, and this output block is selected from the pre-decision block of some honest players, which satisfies the validity property. Thus, \protocol solves the Fully Accountable Agreement problem. \qed 
\end{proof}

\section{Proof of Lemma~\ref{lem:lemma5}}\label{app:lemma5}
\begin{lemma}[Lemma~\ref{lem:lemma5}]
Suppose attackers are able to launch up to $m$ isolation attacks, each of which succeeds independently with probability $p$. The probability of a successfully $\zeta$-uple spending attack, where $\zeta\geq 3$, is upper bounded by $\varepsilon = \left(e/ \lambda\right)^3$.
\end{lemma}

\begin{proof}
We consider a network as a disk where honest players are distributed uniformly as depicted in Figure~\ref{fig:fig5} and where a successful isolation attack consists of delaying all messages on all edges crossing one radius of the disk. In order to create $\zeta$-uple spending attack the adversary must launch $\zeta$ successful isolation attacks.\\
Let $\zeta = \sum_{i=1}^{m} X_i$, where $X_i \perp X_j \quad \forall\, i \neq j,\ i, j \in \{1, \ldots, m\}$ and $\mathbb{P}(X_i)=p$ which is equal to the probability of a successful isolation attack. In other words, $X_i$ is the event of success in the $i$-th Bernoulli trial, where each trial succeeds with probability $p$ and the sum of successful trials is equal to $\zeta$. Let $\mu = \mathbb{E}(\zeta) = \sum_{i=1}^{m} p_i = m \cdot p$. Then, according to Chernoff bounds~\cite{motwani1995randomized}, we have:

\begin{equation*}
    \mathbb{P}\left(\zeta \geq (1+\delta)\mu\right) \leq \left( \frac{e^{\delta}}{(1+\delta)^{1+\delta}} \right)^{\mu} \quad \text{for all } \delta > 0.
\end{equation*}
According to Lemma~\ref{lem:lemma1}, we have $2\leq \zeta\leq \left\lfloor\frac{3n-3(f+k)}{2n-3(f+k)}\right\rfloor$, and $m$ is bounded,
which means attackers cannot establish isolation attacks forever.
Consequently, we have:
\begin{gather}
\mathbb{P}(\zeta \geq 3) \leq \left( \frac{e^{\frac{3}{m \cdot p}-1}}{(\frac{3}{m \cdot p})^{\frac{3}{m \cdot p}}} \right)^{m \cdot p}=e^{3-\frac{3}{\lambda}} \cdot \lambda^{-3}<\varepsilon= \left(\frac{e}{\lambda}\right)^3\quad \text{for } \lambda=\frac{3}{m\cdot p}=\frac{3}{\mu}. \nonumber\\
\end{gather}
As a result, $\mathbb{P}(\zeta \geq 3)<\left(e/ \lambda\right)^3$. \qed  
\end{proof}

\section{Proofs of Theorems~\ref{Theorem:the4}--\ref{Theorem:the5}} \label{app:the45}

In this section, we define the Benefit-of-Betraying, Cost-of-Betraying and Cost-of-Cheating to prove that \protocol is a $\epsilon$-$(k,t)$-robust equilibrium. To this end, we make an extra assumption that there is always at least one Byzantine attacker among $n$ players. This assumption allows us to simplify the proof by guaranteeing that we can reward betrayers with the stake of at least one attacker.

\begin{lemma}\label{lem:lemma2}
If we are in the presence of a betraying coalition, then the \textit{Benefit-of-Betraying} $\ms{BoR}$, which is the reward that increases a player's utility function when this player betrays and does not unlock its reward, is at least $\ms{BoR}=\frac{f\cdot s}{k}$, where $f$ is the number of detected Byzantine attackers.
\end{lemma}
\begin{proof}
By examination of Algorithm~\ref{alg:P}, betrayers received the reward at lines~\ref{line:BoR2} and~\ref{line:BoR}. When an attacker is identified and slashed, its stake is transferred into $\ms{gwallet}$ at line~\ref{line:StoGlobalWallet} and is immediately distributed between valid betrayers at line~\ref{line:BoR}. Moreover, the reward is divided among initial valid betrayers, and each time a new valid betrayer joins, a portion is reclaimed from the existing betrayers and redistributed to the new one at line~\ref{line:BoR2}. We should split the reward equally among betrayers; otherwise, because the network is partially synchronous, players may not have the same order of valid betrayers, and the blockchain could experience inconsistencies. The minimum reward occurs when all rational players betray; consequently, the stake of the other $f$ attackers is split among all $k$ betrayers. Because each $f$ attackers staked $s$ initially, we conclude that the $\ms{BoR}$ per betrayer is at least $\frac{f\cdot s}{k}$.\\
It is worth noting that if a rational player betrays during the commit-reveal phase but is not identified as a valid betrayer due to partial synchrony, then the remaining betrayers receive a larger share of the reward. As a result, the overall expected utility remains unchanged. \qed 
\end{proof}

\begin{lemma}\label{lem:lemma3}
If we are in the presence of a betraying coalition, then the \textit{Cost-of-Betraying} $\ms{CoR}$, which is the decrease in a player's utility function when this player betrays, is $CoR=0$. 
\end{lemma}
\begin{proof}
Rational players who join the coalition initially send conflicting messages. Then, they reveal their proofs of fraud and their slashed stake, as shown at line~\ref{line:CoC} of Algorithm~\ref{alg:P}, will be returned to them as their liquidity, as shown at line~\ref{line:BoR} of Algorithm~\ref{alg:P}. As a result, the cost of betraying is nil. \qed 
\end{proof}

\begin{lemma}\label{lem:lemma4}
If we are in the presence of a betraying coalition, then the \textit{Cost-of-Cheating} $\ms{CoC}$, which is the amount of stake slashed that decreases a player's utility function when this player cheats, gets caught and does not unlock its stake, is $CoC = s$.
\end{lemma}
\begin{proof}
Since the \protocol protocol has the accountability property, all attackers will be identified according to Lemma~\ref{lem:identify_attackers_betrayers} and slashed eventually, as shown at line~\ref{line:CoC} of Algorithm~\ref{alg:P}. \qed 
\end{proof}

\begin{theorem}[Theorem~\ref{Theorem:the4}] 
If $\frac{s}{\ell} > \frac{k\cdot p^2}{f+k}$, then the \protocol protocol is an $\epsilon$-$(k,t)$-robust equilibrium against $\zeta$-uple spending attack.
\end{theorem}

\begin{proof}
Let \textit{Benefit-of-Cheating} $\ms{BoC}$ be the benefit of cheating for a rational attacker after a successful attack. To prove that our protocol is a $(k,t)$-robust equilibrium, we modify equation~\ref{eq:eq1} accordingly:
\begin{equation*}
BoR-CoR>BoC-CoC. 
\end{equation*}
Let $Y_1$ be the event of a double-spending attack which occurs with the probability $p^2$ and $Y_2$ be the event of a $\zeta$-uple spending attack for $\zeta \geq 3$ which occurs with upper-bounded probability $\varepsilon$. Consider $Y_1$ and $Y_2$ are mutually exclusive events such that $Y_1 \cap Y_2 = \varnothing$. In a $\zeta$-uple spending attack, attackers successfully execute the same transaction $\zeta$ times, multiplying their initial liquidity by a factor of $\zeta$ and thus profiting by one liquidity amount. By considering the complexity of isolating a network, the benefit of cheating is:
\begin{equation*}
BoC=\mathbb{E}\left[(\zeta-1)\cdot \ell \right]=\mathbb{E}\left[(\zeta-1)\right] \cdot \ell.
\end{equation*}
Where
\begin{align*}
\mathbb{E}[\zeta-1] &= \mathbb{E}[\zeta-1 \mid Y_1] \cdot \mathbb{P}(Y_1) + \mathbb{E}[\zeta-1 \mid Y_2] \cdot \mathbb{P}(Y_2)\leq p^2+(\mu+c)\cdot \varepsilon\\
&<p^2+\left(\frac{3}{\lambda}+c \right)\cdot \left(e/ \lambda\right)^3=p^2+\mathcal{O}(1/{\lambda}^3)  \quad \text{for }c\text{, a positive constant}.
\end{align*}
As a result:
\begin{equation*}
BoC < \ell \cdot p^2 + \mathcal{O}(1/{\lambda}^3)=\ell \cdot p^2+\epsilon  \quad \text{for }\frac{s}{\ell}\gg \epsilon.
\end{equation*}
By~\Cref{lem:lemma2,lem:lemma3,lem:lemma4}, to satisfy the condition $BoR-CoR +\epsilon> BoC-CoC$ and have \protocol reach a $\epsilon$-$(k,t)$-robust equilibrium despite $\zeta$-uple spending attacks, 
the following stake over liquidity ratio must be satisfied:

\begin{gather*}
\frac{s}{\ell} > \frac{k\cdot p^2}{f+k}.
\end{gather*} \qed 
\end{proof}

\begin{theorem}[Theorem~\ref{Theorem:the5}]
If $\frac{s}{\ell} > p^2$, then the \protocol protocol is an $\epsilon$-$(k,t)$-robust equilibrium that protects against $\zeta$-uple spending attack and bribery. 
\end{theorem}

\begin{proof}
Let \textit{Benefit-of-Cheating} $\ms{BoC}$ be the benefit of cheating for a rational attacker after a successful attack. To prove that our protocol is a $(k,t)$-robust equilibrium, we modify equation~\ref{eq:eq1} accordingly:
\begin{equation*}
BoR-CoR>BoC-CoC. 
\end{equation*}
Let $Y_1$ be the event of a double-spending attack which occurs with the probability $p^2$ and $Y_2$ be the event of a $\zeta$-uple spending attack for $\zeta \geq 3$ which occurs with a probability upper-bounded by $\varepsilon$. Consider $Y_1$ and $Y_2$ are mutually exclusive events such that $Y_1 \cap Y_2 = \varnothing$. By considering the complexity of isolating a network, the benefit of cheating is:
\begin{equation*}
BoC=\mathbb{E}\left[(\zeta-1)\cdot \ell\cdot \left(1 + \frac{f}{k} \right) \right].
\end{equation*}
Where
\begin{align*}
\mathbb{E}[\zeta-1] &= \mathbb{E}[\zeta-1 \mid Y_1] \cdot \mathbb{P}(Y_1) + \mathbb{E}[\zeta-1 \mid Y_2] \cdot \mathbb{P}(Y_2)\\
&\leq p^2+(\mu+c)\cdot \varepsilon<p^2+\left(\frac{3}{\lambda}+c \right)\cdot \left(e/ \lambda\right)^3=p^2+\mathcal{O}(1/{\lambda}^3)  \quad \text{for }c=\text{constant}.
\end{align*}
As a result:
\begin{equation*}
BoC < \ell\cdot \left(1 + \frac{f}{k} \right) \cdot p^2 + \mathcal{O}(1/{\lambda}^3)=\ell\cdot \left(1 + \frac{f}{k} \right) \cdot p^2+\epsilon  \quad \text{for }\frac{s}{\ell}\gg \epsilon.
\end{equation*}
To satisfy the condition where $BoR-CoR > BoC-CoC$ and have a $\epsilon$-$(k,t)$-robust equilibrium, taking into account~\Cref{lem:lemma2,lem:lemma3,lem:lemma4}, we conclude that the \protocol protocol remains secure against $\zeta$-uple spending attacks under the specified condition:
\begin{gather*}
\frac{s}{\ell} > p^2.
\end{gather*} \qed 
\end{proof}

\end{document}